\documentclass[a4paper,10pt]{iopart}
\usepackage[utf8]{inputenc}
\usepackage{color}
\usepackage{graphicx}
\usepackage{amssymb}
\usepackage[bookmarks, bookmarksopen, bookmarksnumbered, colorlinks]{hyperref}

\newcommand{\Iin}{{\ensuremath{I_{\rm in}}}}

\newcommand{\dd}{\ensuremath{\textrm{d}}}

\newcommand{\DU}[2]{\ensuremath{_{#1}^{\phantom{#1} #2}}}

\newcommand{\RealNum}{\ensuremath{{\bf R}}}

\newcommand{\beq}{\begin{equation}}
\newcommand{\eeq}{\end{equation}}
\newcommand{\bea}{\begin{eqnarray}}
\newcommand{\eea}{\end{eqnarray}}
\newcommand{\bit}{\begin{itemize}}
\newcommand{\eit}{\end{itemize}}
\newcommand{\bfi}{\begin{figure}}
\newcommand{\efi}{\end{figure}}
\newcommand{\bfic}{\begin{figure*}}
\newcommand{\efic}{\end{figure*}}
\newcommand{\bce}{\begin{center}}
\newcommand{\ece}{\end{center}}
\newcommand{\bt}{\begin{table}}
\newcommand{\et}{\end{table}}
\newcommand{\btb}{\begin{tabular}}
\newcommand{\etb}{\end{tabular}}

\newcommand{\eff}{\ensuremath{{\it eff}}}
\newcommand{\tot}{\ensuremath{{\it tot}}}
\newcommand{\Mean}[1]{\ensuremath{\left\langle #1 \right\rangle}}
\newcommand{\eeff}[1]{\ensuremath{ #1_{\it eff}}}
\newcommand{\BarPhi}{\Mean{\Phi}}
\newcommand{\disc}{\ensuremath{{\rm disc}}}
\newcommand{\diam}{\ensuremath{{\rm diam}}}

\newcommand{\vol}{\ensuremath{\textrm{vol}}}
\newcommand{\Rr}{\ensuremath{\textbf{R}}}

\newtheorem{theorem}{Theorem}[section]
\newtheorem{lemma}[theorem]{Lemma}

\newenvironment{proof}[1][Proof]{\begin{trivlist}
\item[\hskip \labelsep {\bfseries #1}]}{\end{trivlist}}

\newenvironment{remark}[1][Remark]{\begin{trivlist}
\item[\hskip \labelsep {\bfseries #1}]}{\end{trivlist}}

\newcommand{\qed}{\nobreak \ifvmode \relax \else
      \ifdim\lastskip<1.5em \hskip-\lastskip
      \hskip1.5em plus0em minus0.5em \fi \nobreak
      \vrule height0.75em width0.5em depth0.25em\fi}

\begin{document}

\title{Backreaction and  continuum limit in a closed universe filled with black holes}

\author{Miko\l{}aj Korzy\'nski}
\address{
Center for Theoretical Physics, Polish Academy of Sciences \\
Al. Lotnik\'ow 32/46 \\
02-668 Warsaw \\
Poland}
\ead{korzynski@cft.edu.pl}

\begin{abstract}
We discuss the continuum limit of the initial data for a vacuum, closed cosmological model with black holes as the only sources of the gravitational field. The model we consider is an exact solution of the constraint equations and represents a vacuum universe with a number of black holes placed on a spatial slice of $S^3$ topology considered at the moment of its largest expansion when the black holes are momentary at rest. 
We explain how and under what conditions the FLRW metric arises as the continuum limit when the number of black holes contained in the model goes to infinity. We also discuss the
relation between the effective cosmological parameters of the model, inferred from the large scale geometry of the spacetime, and the masses of individual black holes. In particular, we prove an estimate for the difference between the total effective mass of the system and the sum of the masses of all black holes, thus quantifying  the effects of the inhomogeneities in the matter distribution or the cosmological backreaction.
\end{abstract}

 \section{Introduction}

 The Einstein's field equations relate the Einstein tensor, a part of the curvature of the spacetime metric, to the stress--energy tensor of the matter. They are non--linear and covariant, the latter feature being the expression of the coordinate system invariance of gravity. 
In astrophysical applications of general relativity we often need to solve them in a situation in which the matter distribution has  a complicated, lumpy or even discrete structure on small (microscopic) scales, but a simpler and more uniform one on a larger (macroscopic) scale. The standard method of dealing with this problem  is to apply the Einstein's equations to an idealized large--scale metric, called the background, which is supposed to represent correctly the properties of the spacetime on the macroscopic scale. It is assumed that the influence of the individual microscopic sources has been removed via an averaging or coarse--graining procedure. The averaging scale is understood to be mesoscopic, i.e. to lie somewhere between the micro-- and macroscopic. The stress--energy tensor on the right hand side of the  equations is then assumed to be  the simple sum  of contributions from all small--scale sources located within a mesoscopic volume. 

The main advantage of the  procedure outlined above is that we effectively replace a complicated, discrete structure of matter and the metric tensor with a simpler, continuous one. This is obviously an approximation which can be justified as follows: assuming that the Einstein's equations of gravity hold exactly at smallest scales we can perform the coarse--graining of the metric tensor $g_{\mu\nu}$ and substitute the resulting background metric
$g^{\eff}_{\mu\nu}$ to the Einstein equations. Due to the non--linear nature of the Einstein's equations and possibly of the coarse--graining procedure the resulting effective stress--energy tensor $T^{\eff}_{\mu\nu}$ will in general differ from the effect of the coarse--graining of the local stress--energy tensor $\Mean{T_{\mu\nu}}$ \cite{Ellis:2011}. 
If  the microscopic scale is small enough, we may expect the non--linear effects of GR to play little role at that scale and below and  
therefore  $\Mean{T_{\mu\nu}}$ will essentially be the sum over the local contributions from small--scale sources. The difference between $\Mean{T_{\mu\nu}}$ and  $T^{\eff}_{\mu\nu}$ is called in this
context the backreaction. The continuum approximation is then based on the assumption that the backreaction term is negligible in comparison to  the effective  stress--energy tensor if the ratio between the two scales is large enough. 

Although  this approximation seems reasonable, it would be very useful to provide its mathematically rigorous justification and, more importantly, draw precisely its limits of applicability. In the recent years the problem has attracted attention
of the researchers in the context of cosmology \cite{Buchert:2011yu, Wiltshire:2013wta, Bolejko:2008yj, Ishibashi:2005sj, Green:2010qy, Brown:2013usa, Coley:2010hs, Zalaletdinov:2008ts, Wiltshire:2011vy, Andersson:2013uaa, Rasanen:2009uw, Brown:2009tg, Reiris:2007gb, Korzynski:2009db}. Obviously in the cosmological setting we assume the existence of a homogeneity scale, which may be as large as 2000Mpc \cite{Horvath:2013kwa},   and a complicated, multiscale structure below that. The  problem of backreaction takes  the form of the question of applicability of the Friedmann equations with ordinary matter content to the observational cosmology. In particular, the big question is whether or not the backreaction terms might mimic the dark energy whose existence is inferred from, among other things, the supernovae  observations.

The general problem of justification of the continuum approximation can be broken into several smaller, related problems. In this paper we will be concerned with two of those:

\emph{The continuum limit problem:} Assume that the ratio between the macroscopic scale and the microscopic one goes to infinity. This means that the size of individual lumps goes to 0 when compared to the typical macroscopic scale. Under what conditions does the physical metric tensor approach the coarse--grained, background metric and in what mathematical sense? Obviously we should not expect the convergence to occur near very compact matter sources, where the metric is strongly influenced by their presence, but rather in the faraway region.

\emph{The backreaction problem:} Can we evaluate the backreaction effects given the macroscopic metric and matter distribution and the statistical properties of the distribution of microscopic sources? In particular, can we say under what conditions the backreaction effects vanish as we pass to the continuum limit in the sense described above? The problem is especially interesting and difficult if the small--scale sources do not contain any ordinary matter themselves, but rather take the form vacuum black hole solutions of Schwarzschild or Kerr-like type. The spacetime is vacuum everywhere in this case, so the average of the local $T_{\mu\nu}$ vanishes  and the only contributions to the 
macroscopic stress-energy tensor may come from summing over the discrete sources. The continuum limit corresponds to the limit of the number of black holes going to infinity.

These problems are very hard to solve if we approach them in full generality and it is necessary to make simplifications here.  One possibility is
to consider only the initial data for the Einstein's equations, neglecting the dynamics completely. 
Recall that the initial data in the form of a manifold $M^3$ with a Riemannian metric $q_{ij}$ and a symmetric tensor $K^{ij}$
are related to 4 of 10 independent components of the stress--energy tensor via the Hamiltonian and vector constraint equations \cite{Baumgarte:1998te}:
\bea
 R + K^2 - K_{ij}\,K^{ij} &=& 16\pi T_{nn} \nonumber\\
 D_i K^{ij} - D^j K &=& 8\pi T^{i}_n, \nonumber
\eea 
$n$ denoting the normal component with respect  to the constant time 3--surface. 
We can perform the coarse--graining at the level of the initial data starting from $q_{ij}$ and $K^{ij}$, obtaining the effective
energy density $T_{nn}$ and momentum density $T_{ni}$. Note that the rest of the stress energy tensor, and consequently the rest of the backreaction, is inaccessible without considering the dynamics of the system. However the energy  and momentum densities and 
their corresponding backreaction terms 
are themselves of great physical interest. 
 
Secondly, instead of deriving an exact expression for backreaction we may consider a more modest goal of proving an inequality estimating the backreaction from above. The inequality should of course be useful in the sense that it implies the vanishing of backreaction in the continuum limit under certain physically motivated assumptions, thereby justifying the ordinary continuum approximation.

In this paper we will discuss both problems in a special but physically interesting case of the initial data for a  vacuum 3--sphere with
black holes. This family of time symmetric initial data  is a counterpart of a closed FLRW model at the moment of its largest expansion in which the homogeneous matter distribution has been replaced with  an arbitrary number of Schwarzschild--like black holes located on a 3--sphere. This type of initial data has been studied extensively in the special case of regular 
black hole lattices, i.e. arrangements of
a fixed number of black holes of equal mass placed in such a way that the resulting initial data has the biggest possible discrete  group of isometries  \cite{RevModPhys.29.432, springerlink:10.1007/BF01889418, Bentivegna:2012ei, Clifton:2012qh,Clifton:2013jpa}, see also \cite{Yoo:2012jz, Yoo:2013yea, Bentivegna:2013jta} for  regular configurations constructed on a flat space rather than a sphere. In \cite{Clifton:2012qh} it was noted that as the number of black holes $N$ grows, the metric tensor in the region between the black holes becomes increasingly round. Unfortunately only 6 arrangements of this type are possible, with maximally 600 black holes, so these configurations are insufficient if we are interested in probing the continuum limit for $N\to\infty$. Therefore here we will be concerned with absolutely arbitrary arrangements of black holes on $S^3$. 

The main result of this paper consists of two related inequalities. The first one estimates  the order of convergence to the continuum limit of the initial data described above, while the other one bounds the only backreaction  effect possible here, i.e. the non--additivity of the black hole masses. 
More precisely, the first inequality provides bounds on the difference between the local metric and the properly defined coarse grained metric at any point of the manifold. The latter bounds the difference between the effective total mass of the model, defined via the  Einstein's equations applied to the coarse--grained metric and extrinsic curvature, and the sum of the ADM masses of the black holes. Both bounds depend on certain parameters measuring the properties of the  arrangement
of the black holes, which in this case play the role of the microscopic matter distribution. We also present a construction of an infinite sequence of black hole configurations  for which the two inequalities can be used to prove that the continuum limit exists in a sense we explain in detail  and that asymptotically the backreaction vanishes for large number of black holes.

The proofs of the inequalities constitute the most involved part of the paper. The difficult part in both cases is 
estimating  terms of the form of the difference between the volume average of a function on $S^3$ and the weighted average of its values 
taken over a finite set of points on the 3--sphere. The problem of such estimates has a long history in mathematics,
dating back to H. Weyl \cite{Weyl1916}.  It turned out however that  none of the results we have found in the literature is directly applicable to our case. Therefore our methods of proof, although partially based on known results, are new. 

The paper is organized as follows: in the next section we review the construction of the vacuum initial data on $S^3$ with
black holes and discuss their basic properties. We then propose a way of fitting a closed FLRW model at the moment of 
the largest expansion to the given initial data. This can be thought of as a special method to coarse--grain the inhomogeneous 
spatial metric tensor, yielding always a completely homogeneous one. In Section 3 we state the two main theorems of the paper right after 
introducing the necessary mathematical apparatus and then prove them  in Sections 4 and 5. In Section 6 we discuss the construction of an infinite sequence of black hole configurations in which the black holes cover the sphere in a uniform manner. We prove that for those configurations the continuum limit is achieved in an appropriate sense and the backreaction vanishes as the number of black holes $N$ goes to infinity.
Applying the inequalities proved in the previous section we  discuss the properties of the initial data for very large $N$.
In Section 7 we present the result of numerical investigation of the properties of the 6 regular configurations as well as the configurations
from the sequence constructed in Section 6. We summarize the results and spell out the conclusions in Section 8. In the Appendix we include the proofs of several technical results we have used in the paper.

\section{Static initial data on $S^3$}

Consider a three--dimensional sphere with the standard metric
\bea
q_0 = \dd \lambda^2 + \sin^2\lambda\left(\dd \vartheta^2 + \sin\vartheta^2\,\dd\varphi^2\right).\nonumber
\eea
We construct vacuum initial data using the Lichnerowicz--York prescription (see \cite{lichnerowicz:1944, York:1971hw}), i.e. assume the physical  three-metric $q$ to be related to $q_0$ by a conformal rescaling
\bea
q = \Phi^4\,q_0. \label{eq:conformal}
\eea
We assume the initial data to be time-symmetric
\bea
K_{ij} = 0, \label{eq:Kzero}
\eea
which means that we are looking at this universe at the moment of its largest expansion. The vector constraint equation is satisfied automatically, while the scalar constraint reduces to the Lichnerowicz--York equation for the conformal factor. In this case it turns out to be linear:
\begin{eqnarray}
\left(\Delta - \frac{3}{4}\right) \Phi = 0, \label{eq:ScalarConstraint}
\end{eqnarray}
$\Delta=q_0^{ij} \nabla_i \nabla_j$ denotes here the ''round'' Laplacian on $S^3$.

Equation (\ref{eq:ScalarConstraint}) does not have regular solutions, but it nevertheless has a Green's function
\begin{eqnarray}
\Psi(x) = \frac{1}{\sin\frac{\lambda}{2}} \label{eq:Greens}
\end{eqnarray} 
with a single singularity at $\lambda=0$ where it diverges like $\frac{2}{\lambda}$. We can use $\Psi(x)$ to construct solutions with a finite number of punctures, i.e. singularities located at selected points $x_n$: 
\begin{eqnarray}
\Phi(x) = \sum_{n=1}^N \alpha_n\frac{1}{\sin \frac{\Lambda(x,x_n)}{2}}, \label{eq:PhiDefinition}
\end{eqnarray}
where $\Lambda(\cdot,\cdot)$ denotes the geodesic distance with respect to
 $q_0$. The positive coefficients $\alpha_n$ will be referred to as the mass parameters. It is tempting to identify them with the bare masses of  the punctures, but note that their dimension is $L^{1/2}$ instead of  $L$.
  
 We will denote by $A$  a non--empty set of  pairs  $(x_i, \alpha_i)$ consisting of a puncture and its positive mass parameters. Such sets, defining via (\ref{eq:PhiDefinition}) the conformal factor, will be referred to as configurations. We introduce a notation for the sum of all mass parameters in a given configuration:
 \bea
 \alpha = \sum_{i=1}^{N} \alpha_i.\nonumber
\eea 
Note that equation (\ref{eq:ScalarConstraint}), being linear, can be interpreted distributionally. In this case solution (\ref{eq:PhiDefinition}) 
satisfies  the distributional equation
\begin{eqnarray}
\left(\Delta - \frac{3}{4}\right) \Phi = -8\pi\sum_{n=1}^N \alpha_n\delta(x_n,x), \label{eq:ScalarConstraintDistr}.
\end{eqnarray}
$\delta(x_n,x)$ being the Dirac delta on $S^3$ at $x_n$.

Each of the punctures corresponds to an asymptotically flat region of the spacetime. For $N=1$ the resulting metric $q$ is isometric to the flat space $\RealNum^3$. For $N=2$ it is isometric to the standard $t=0$ section of a Schwarzschild spacetime. For $N \ge 3$ the resulting 3--manifold consists of $N$ asymptotically flat ends connected  through minimal surfaces 
to a central region (see figure \ref{fig:geometry}). In analogy with the Schwarzschild solution we will refer to them as the black holes. The ADM mass measured at each of the asymptotically flat end
$x_n$ is given by
\bea
M_n = 4\alpha_n\,\Phi_{\rm reg}(x_n),\nonumber
\eea
where $\Phi_{\rm reg}(x) = \sum_{m=1,m\neq n}^{N} \alpha_m\frac{1}{\sin \frac{\Lambda(x,x_m)}{2}}$ is the
regular part of $\Phi(x)$ at $x_n$. This is equivalent to the following formula
\bea
M_n = 4\alpha_n\sum_{m=1,m\neq n}^{N} \frac{\alpha_m}{\sin \frac{\Lambda(x_m,x_n)}{2}} \label{eq:ADMmass}.
\eea
The expression is quadratic in $\alpha_i$'s, so its dimension is indeed $L$ as it should be. Note however that it depends not only on $\alpha_i$, but also on mass parameters of all other punctures as well as on their positions on $S^3$. 
\bfi
\bce
\includegraphics[trim=2cm 4cm 2cm 3cm,clip,width=0.45\textwidth]{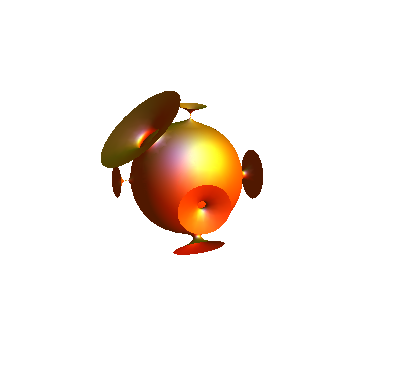}
\caption{Geometry of the closed $S^3$ initial data with 6 black holes. Trumpets correspond to asymptotically flat ends.
\label{fig:geometry}}
\ece
\efi
The sum of ADM masses of  all punctures, which we will call the \emph{total mass} of the configuration, is given by
\bea
M_{ \tot} = \sum_{n=1}^N \sum_{m=1,m\neq n}^{N} \frac{4\alpha_m\alpha_n}{\sin \frac{\Lambda(x_m,x_n)}{2}}.\label{eq:Mtot}
\eea
Just like $M_n$ it depends on the positions of the black holes. 

These initial data have been first discussed by Wheeler and Lindquist in \cite{RevModPhys.29.432} and later by many authors. The research so far  has mostly focused on a special type of configurations called \emph{regular}. These are the maximally symmetric arrangements of identical black holes, forming a finite regular lattice on a $S^3$. If we consider the standard embedding of $S^3$ into $\Rr^4$, the punctures in these configurations are located at the vertices of a
regular polytope (a 4--dimensional generalization of a platonic solid) inscribed into the 3--sphere. It turns out that there are only
6 configurations of this type, with $N=5$, 8, 16, 24, 120 and 600. They will be denoted in this paper by $R_5$, $R_8$, $R_{16}$ $R_{24}$, $R_{120}$ and $R_{600}$ respectively. 

Wheeler described the regular configuration with 5 black holes in detail in \cite{springerlink:10.1007/BF01889418}.
The interest in these models has been revived recently in the context of the backreaction problem: in \cite{Clifton:2012qh} Clifton, Rosquist and Tavakol described in detail and discussed all 6 regular configurations focusing on the backreaction effects. In \cite{Bentivegna:2012ei} the results of numerical evolution of $R_8$  were presented and compared with the closed FLRW model with dust. Finally in \cite{Clifton:2013jpa} Clifton, Gregoris, Tavakol and Rosquist discussed the exact evolution of the regular configurations along certain geometrically distinguished curves, taken to arbitrary large times.

Let us stress here that the linearity of equation (\ref{eq:ScalarConstraint}) and the superficial resemblance of (\ref{eq:ScalarConstraintDistr}) to the Poisson equation for the Newtonian gravitational potential does not mean that we are considering here a simplified, linearized version of general relativity. The pair $(q_{ij}, K^{ij})$, given by the non--linear equation (\ref{eq:conformal}) and equation (\ref{eq:Kzero}), is an exact solution of the full  constraint equations and therefore captures also the non--linear effects of Einstein's theory of gravity. 
This is evident in equations (\ref{eq:ADMmass}) and (\ref{eq:Mtot}), where both the ADM masses and the total mass turn out to be quadratic in the mass parameters. This is unlike the linearized GR or Newtonian gravity, where the mass of a point source depends linearly on the coefficient of the  order of $1/r$ in the expansion of the metric or the Newtonian potential. 
Even more striking is the inherently non--local  nature of both types of mass: although adding a new black hole to the configuration amounts to simply adding 
another term to (\ref{eq:PhiDefinition}) without affecting the rest of $\Phi$, the ADM mass of every other black hole changes because of its presence. This is a clear indication of the non--linear effects of gravity.

\subsection{FLRW fitting and effective parameters}

It is intuitively clear that the initial data described above are the discretized counterpart of the closed dust FLRW model at the moment of its largest expansion. However,
in order to make a direct comparison we need to propose a way to single out one FLRW model whose constant time slice resembles our model in large scales. It is obvious that the FLRW counterpart must be a closed universe ($k=1$) observed at the moment of its largest expansion. Therefore we only need now to measure the effective size of the model, i.e. propose the 
counterpart of the $a(t)$ parameter in a more general setting.  The measure should ideally be non--local, capturing the large scale properties of the whole model rather than the local properties of a small region of the spacetime. The total 3--volume of the constant time slice is 
a natural candidate in a closed FLRW model. Unfortunately it cannot be used in our case because it is infinite, due to of the infinite volume inside the black holes. In \cite{Clifton:2009jw, Bentivegna:2012ei, Clifton:2013jpa}, where only regular lattice configurations are considered, the size is defined by matching the size of the lattice cell to a corresponding lattice cell in the FLRW model. This is not possible in
the general case where no symmetries are present and no curves are geometrically distinguished. 

In this paper we will use  
the volume average of $\Phi$, taken over the whole 3--sphere with respect to the unphysical, round volume form $\eta = \sin^2\lambda\,\sin\vartheta\,\dd\lambda\wedge\dd\vartheta\wedge\dd\varphi$, as the measure of scale. Recall that the
total volume of a 3--sphere is $2\pi^2$. We define
\bea
\Mean\Phi = \frac{1}{2\pi^2}\int_{S^3} \Phi\,\eta.\nonumber
\eea
Since the singularities in $\Phi$ are of the order of $\lambda^{-1}$, the integral above is convergent.
Comparing (\ref{eq:conformal}) with the spatial closed FLRW metric $a(t)^2\,q_0$ we propose the following definition of 
the effective scale factor $a$:
\bea
\eeff a = \BarPhi^2\nonumber \label{eq:eeffa}
\eea
and the coarse--grained metric
\bea
\Mean{q} = \eeff a^2\,q_0 = \Mean{\Phi}^4\, q_0.\nonumber
\eea

Since the initial data is static, we may safely assume that $\Mean{K_{ij}} = 0 $ and consequently $\dot \eeff a = 0$.
$\eeff a$ and $\dot \eeff a$ suffice to define the effective Ricci scalar and energy density of the configuration via the first Friedmann equation
\bea
\left(\frac{\dot \eeff a}{\eeff a}\right)^2 + \frac{\eeff R }{6}= \frac{8\pi\eeff \rho}{3} \nonumber
\eea
and the relation between the scale and the curvature in the closed FLRW universe
\bea
\eeff{R} = \frac{6k}{\eeff{a}^2} =\frac{6}{\eeff{a}^2}.\nonumber
\eea
Taken together they yield
\bea
\eeff \rho = \frac{3}{8\pi}\,\Mean \Phi^{-4}.\nonumber
\eea
The momentum density $T_{ni}$  in this model vanishes identically, so the only interesting backreaction effect arises
as the difference between the integral of $T^{\eff}_{nn} = \eeff \rho$ over a mesoscopic domain and the sum of the ADM masses of black holes contained inside. 
We define the \emph{effective mass} of the model as the total mass of the corresponding FLRW model, i.e.
\bea
\eeff M = 2\pi^2 \eeff a^3\eeff \rho = \frac{3\pi}{4} \Mean\Phi^2.\nonumber
\eea
The corresponding FLRW model is completely homogeneous, so the effective total mass contains exactly the same information about the matter content as the effective
local energy density. We have found the former more convenient for the purposes of this paper and  in the next section we will use $\eeff M$ rather than $\eeff \rho$ to quantify the backreaction effects.

Let us note that for a given configuration $A$ the mean conformal factor $\Mean\Phi$ can be expressed in terms of the mass parameters in the following way: the Green's function $\Psi$ is integrable despite having a singularity and its average satisfies
\bea
\Mean{\Psi} = \frac{16}{3\pi}, \label{eq:avPsi}
\eea
so
\bea
\Mean\Phi =  \frac{16}{3\pi}\sum_{n=1}^N \alpha_i = \frac{16}{3\pi}\alpha. \label{eq:avPhi}
\eea
This leads to the expression for the effective mass in terms of $\alpha$:
\bea
\eeff M = \frac{64}{3\pi}\,\alpha^2. \label{eq:Meff}
\eea
It is quadratic in $\alpha_n$ just like $M_{\tot}$,  however it depends only on the values of the mass parameters and is independent
of the punctures' positions.

The proposed method of constructing the coarse--grained metric, basing on volume averaging over an unphysical volume form, may seem entirely artificial, we will see however that it produces a reasonable results in the  sense of giving the correct expression for the metric tensor in the continuum limit. Moreover, as we will see, other coarse--graining procedures we could possibly apply here will yield similar results in most cases.

\section{Estimates for the deviation of the conformal factor from the mean  and the backreaction} 

Let $\sigma_\Phi $ denote the relative difference between the average volume form and its value at a point $p$
\bea
\sigma_\Phi(p) = \frac{\left|\Phi(p) - \BarPhi\right|}{\BarPhi}.\nonumber
\eea
We substitute (\ref{eq:avPhi}) and (\ref{eq:PhiDefinition}) and after simple transformations get
\bea
\sigma_\Phi  = \frac{3\pi}{16}\left|\sum_{n=1}^N \frac{\alpha_n}{\alpha}\,\frac{1}{\sin \frac{\Lambda(p,x_n)}{2}} - \frac{16}{3\pi}\right|.\nonumber
\eea
Recall that the volume average of the Green's function (\ref{eq:Greens}) is equal $\frac{16}{3\pi}$, see (\ref{eq:avPsi}). Thus
the previous expression can be written as
 \bea
 \sigma_\Phi =\frac{3\pi}{16}\left|\sum_{n=1}^N \frac{\alpha_n}{\alpha}\,\frac{1}{\sin \frac{\Lambda(p,x_n)}{2}} - \Mean{\frac{1}{\sin \frac{\lambda}{2}} } \right|. \label{eq:rrr}
\eea
The first term in the modulus is the weighted sum of $N$ Green's functions centered at $x_n$ and evaluated at a single point $p$, but since $\Lambda(\cdot,\cdot)$ is symmetric we can turn this interpretation around and consider it as the weighted average of a \emph{single} Green's function centered at $p$ and evaluated at $x_i$'s. $\sigma_\Phi$ acquires thus the interpretation of the difference between a weighted average of values of function $\Psi_p(x) = \left(\sin \frac{\Lambda(p,x)}{2}\right)^{-1} $ and its volume average. Even without further analysis we may expect that if pairs $A$ cover $S^3$ in a uniform way the weighted sum should approximate the integral pretty well and $\sigma_\Phi$ should go to $0$ as $N$ goes to infinity. We introduce a notation for this type of difference
\begin{eqnarray}
G(A,p) = \sum_{n=1}^N \frac{\alpha_n}{\alpha}\,\frac{1}{\sin \frac{\Lambda(p,x_n)}{2}} - \Mean{\frac{1}{\sin \frac{\lambda}{2}} },
\label{eq:Gdef}\end{eqnarray}
so
\bea
\sigma_\Phi  = \frac{3\pi}{16}\left| G(A,p) \right|. \label{eq:r}
\eea

Consider now the relative mass deficit, measuring the backreaction in the normal--normal component of Einstein equations, defined as
\bea
\sigma_M = \frac{\left|\eeff M - M_\tot\right|}{M_\tot}.\nonumber
\eea
We will show that it can be estimated by a similar expression. From (\ref{eq:Meff}) we have 
\begin{eqnarray}
\fl \eeff M =  \frac{64}{3\pi}\sum_{n=1}^N\alpha_n\,\sum_{m=1}^N\alpha_m =\frac{64}{3\pi}\sum_{n=1}^N\alpha_n\sum_{m=1,m\neq n}^N\!\!\!\!\alpha_m + \frac{64}{3\pi}\sum_{k=1}^N \alpha\DU{k}{2}. 
\end{eqnarray}
We have separated out the diagonal terms from the double sum. Once again we substitute (\ref{eq:avPsi}) to get
\bea
\eeff M = 4 \Mean{\frac{1}{\sin \frac{\lambda}{2}}} \sum_{n=1}^N\alpha_n\sum_{m=1,m\neq n}^N\!\!\!\!\alpha_m + \frac{64}{3\pi}\sum_{k=1}^N \alpha\DU{k}{2}\nonumber
\eea
and therefore
\bea
\fl \sigma_M = \frac{1}{M_{\tot}}\left|4 \sum_{n=1}^N\alpha_n \left(\sum_{m=1,m\neq n}^N\frac{\alpha_m}{\sin \frac{\Lambda(x_m,x_n)}{2}} - (\alpha - \alpha_n) \Mean{\frac{1}{\sin \frac{\lambda}{2}}}\right) + \frac{64}{3\pi}\sum_{k=1}^N \alpha\DU{k}{2}\right|.\nonumber
\eea
We substitute the definition of $ M_{\tot}$ from (\ref{eq:Mtot}) to get
\bea
\fl \sigma_M = \left| \frac{3\pi}{16}\sum_{n=1}^N
\left(\frac{\alpha_n}{\alpha} - \frac{\alpha_n^2}{\alpha^2}\right)\left(\sum_{m=1,m\neq n}^N\!\!\frac{\alpha_m}{\alpha- \alpha_n}\frac{1}{\sin \frac{\Lambda(x_m,x_n)}{2}} - \Mean{\frac{1}{\sin \frac{\lambda}{2}}}\right)  + \sum_{k=1}^N \frac{\alpha_k^2}{\alpha^2}\right|.\nonumber
\eea
The term in the second brackets is obviously the difference between two types of average of $\Psi_p(x)$, just like  (\ref{eq:rrr}), but  without the
contribution from the puncture at
$x_n$ and centered at $p=x_n$: 
\bea
&&\sigma_M = \left| \frac{3\pi}{16}\sum_{n=1}^N \frac{\alpha_n}{\alpha}
\left(1 - \frac{\alpha_n}{\alpha}\right) G\left(\widetilde A_n, x_n\right)
+ \sum_{k=1}^N  \frac{\alpha_k^2}{\alpha^2} \right|,\nonumber\\
&&\widetilde A_n = A \backslash \left\{ \left(x_n,\alpha_n\right)\right\}.\nonumber
\eea
From this we obtain easily bounds for  $\sigma_M$   of the form of
\bea
\sigma_M \le  \frac{3\pi}{16}\sum_{n=1}^N \frac{\alpha_n}{\alpha}
\left(1 - \frac{\alpha_n}{\alpha}\right)\,\left| G\left(\widetilde A_n,x_n\right ) \right| + \sum_{k=1}^N  \frac{\alpha_k^2}{\alpha^2} \nonumber
\eea
(we have used here the fact  that $\alpha_n>0$ and, what follows easily, that $\alpha_n < \alpha$). This can be further simplified if we omit the
negative term in the brackets and bound the $\left| G(\widetilde A_n,x_n ) \right|$ by the maximum
taken over all punctures:
\bea
\sigma_M \le   \frac{3\pi}{16}\,\max_{n=1\dots N} 
\left|G\left(\widetilde A_n, x_n\right)\right| + \sum_{k=1}^N \left( \frac{\alpha_k}{\alpha} \right)^2.\nonumber
\eea
Let $\alpha_{\it max} = \max_{i=1\dots N} \alpha_i$. The second term can be bounded  from above in the following way:
\bea 
\sum_{k=1}^N \left( \frac{\alpha_k}{\alpha} \right)^2 \le \frac{\alpha_{\it max}}{\alpha}\,\sum_{k=1}^N \frac{\alpha_k}{\alpha} =
\frac{\alpha_{\it max}}{\alpha}\nonumber,
\eea
so finally
\bea
\sigma_M \le   \frac{3\pi}{16}\,\max_{n=1\dots N} 
\left|G\left(\widetilde A_n, x_n\right)\right| + \frac{\alpha_{\it max}}{\alpha}. \label{eq:s}
\eea
Obviously the second term goes to 0 as $N\to\infty$ if we distribute the mass parameters among black holes evenly. The difference thus depends mainly on the first term.

Finding bounds for expressions of  type (\ref{eq:Gdef}) has a long history in mathematics. It became very important  with the advent of the quasi-Monte Carlo methods of integration, where one tries to approximate the integral by a finite sum of values taken at a series of points $(x_i)$ lying inside the domain of integration \cite{owen2006, Sobol1998103}:
\bea
\int _{[0,1]^d} f \,dx = \sum_{i=1}^N \frac{1}{N}\,f(x_i) + E\nonumber
\eea
In order to estimate the quality of the approximation one would like to estimate the absolute value of the error $E$, which is given by an expression of type (\ref{eq:Gdef}).
The most important tool for this purpose is the Koksma--Hlawka inequality \cite{Hlawka1961}, in which 
the difference $E$ is bounded by the product
of two quantities: the so-called star discrepancy of the series, measuring how well it covers the whole domain of integration and
the Hardy--Krause variation of the function $f$, measuring how much the value of $f$ jumps inside the domain \cite{DrmotaTichy, lmsobol1}:
\bea
|E| \le D^*({x_i})\,V_{H-K}(f). \nonumber
\eea
For appropriately chosen sequences $(x_i)$  $D^*(x_i)$ goes to $0$ as $N\to\infty$ at an explicitly known rate \cite{lmsobol1}.
Unfortunately the Koksma--Hlawka inequality does not work for unbounded functions like (\ref{eq:Greens}), whose Hardy--Krause variation is infinite. 
While many papers have been written on the bounding $E$ if $f$ has a singularity \cite{owen2006, Hartinger2004654, deDoncker2003259, sobol1973}, none of them seems to offer methods applicable in our case, where the domain of integration is spherical, the volume form  is non--standard and the weights in the average may be more complicated than simply $\frac{1}{N}$.
 We have thus developed our own
tools for estimating $\sigma_\Phi$ and $\sigma_M$, based on an improved version of the Koksma--Hlawka inequality, which we will present below. Before we go on however let us note that when the function $f$ has a singularity, the convergence of $E$ to 0 is not guaranteed even if the sampling points cover the domain uniformly. This is because $E$ can always be driven to very large or small values by a single point approaching the singularity too close and driving the weighted sum up or down. Therefore the estimate for $E$ for unbounded $f$ must  have a more complicated structure, involving the dependence on the distance from the singularity to the nearest sampling point.

\subsection{Cap discrepancy function and the modified discrepancy of $A$}
Before we state our main result we must introduce two important tools measuring how evenly a set of punctures
covers the 3--sphere.
Let $B(p,\lambda)$ denote a 3--spherical cap centered at $p\in S^3$, i.e. $B=\left\{x\in S^3| \Lambda(p,x)\le \lambda\right\}$.
We define the \emph{cap discrepancy function} $\disc_p(\lambda)$ as the difference between the properly normalized volume of 
a 3--spherical cap of radius $\lambda$ centered at $p$ and the sum of mass parameters of all punctures contained within the cap:
\bea
\disc_p(\lambda) = \frac{1}{2\pi^2}\vol(B(p,\lambda)) - \sum_{x_i\in B(p,\lambda)} \frac{\alpha_i}{\alpha}. \label{eq:disc}
\eea
The discrepancy function measures how well the volume of a cap measured with respect to the round volume form agrees with the volume measured with respect to a discrete measure centered at $x_i's$ (see \cite{DrmotaTichy} for a general discussion of the notion of sequence discrepancy on a sphere). Obviously the smaller values it takes on the whole $[0,\pi]$ interval, the more evenly the 3--sphere is covered with punctures. We may also expect that the values of $\disc_p(\lambda)$ will decrease as we add more and more punctures in a uniform manner. Note that by definition $\disc_p(\pi) = 0$ and, if $p$ is not a puncture, $\disc_p(0) = 0$.

If we introduce a spherical coordinate system $(\lambda,\vartheta,\varphi)$ centered at $p$, the discrepancy function
takes a simpler form:
\bea
\disc_p(\lambda) = \frac{2}{\pi}\Xi(\lambda) - \sum_{i=1}^N \frac{\alpha_i}{\alpha}\theta(\lambda - \lambda_i),\label{eq:disc2}
\eea 
where $\Xi$ is a short-hand notation for the antiderivative
\bea
\Xi(\lambda) = \int_0^\lambda \sin^2 x\,dx, \label{eq:Xi}
\eea
$\theta$ is the  step function
\bea
\theta(\lambda) = \left\{
\begin{array}{ll}
0 & {\rm for }\qquad \lambda < 0 \\
1 & {\rm for }\qquad \lambda \ge 0
\end{array}\right. \label{eq:step}
\eea
and $\lambda_i$ is the $\lambda$ coordinate of puncture $x_i$.

We would like to define a single number providing the upper bound for the values of the cap discrepancy function for all $p\in S^3$. The standard method here is to take the global supremum over $\lambda$ and $p$, called the \emph{total cap discrepancy} \cite{DrmotaTichy}. This parameter turns out to be insufficient for our purposes and therefore we introduce a stronger version of a global bound for $\disc_p(\lambda)$. Let  $F_D(\lambda)$ be the one--parameter family of functions
\bea
F_D(\lambda) = \left\{
\begin{array}{ll}\frac{4D}{\pi}\sin^2(\lambda + D) &\textrm{for }0 \le \lambda \le \frac{\pi}{2} - D \\
\frac{4D}{\pi} & \textrm{for }\frac{\pi}{2} - D \le \lambda \le \pi\end{array} \right.\nonumber
\eea
(see figure \ref{fig:FD}). The smallest $D$ such that $\disc_p(\lambda) \le F_D(\lambda)$ for all
$0\le\lambda\le\pi$ will be called the \emph{modified cap discrepancy at $p$} and denoted by $E_p$. The supremum of $E_p$ over the whole 3--sphere will be called the \emph{global modified discrepancy} of configuration $A$ and denoted by $E$. Let us stress that
$E$, being independent of the position on the sphere,  is a property  of the configuration $A$ itself. In comparison to the standard total cap discrepancy, given by the supremum of $\disc_p(\lambda)$,  $E$ bounds the values of $\disc_p(\lambda)$ more tightly close to $\lambda=0$, otherwise its properties are very similar. Obviously the smaller $E$,  the more tightly the cap discrepancy function values are bounded and the more evenly the punctures must be spread over the sphere.
\bfi
\bce
\includegraphics[width=0.80\textwidth]{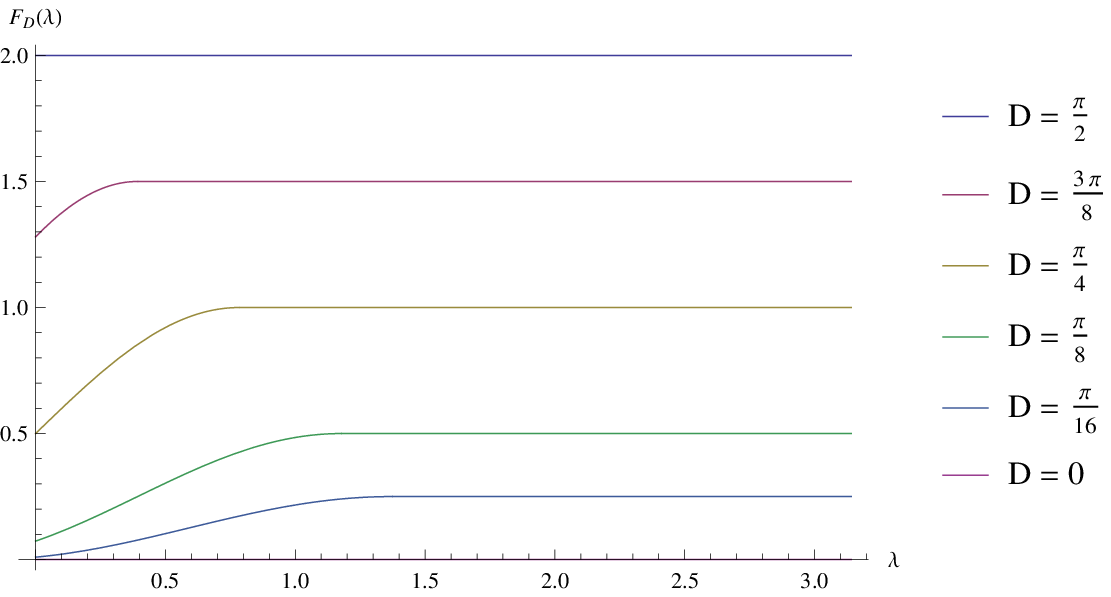}
\caption{Functions $F_D(\lambda)$ for different values of $D$}
\label{fig:FD}
\ece
\efi

\subsection{The main theorems}
We are now ready to formulate two main theorems of the paper.
\begin{theorem} 
Let $p \in S^3 \backslash \left\{x_1,\dots,x_n\right\}$, let $\lambda_{\min} = \min_{n=1,\dots,N} \Lambda(p,x_n)$ and let $E$ be the modified discrepancy of $A$. Let $\varepsilon$ satisfy $0 < \varepsilon \le 1$. Then, assuming that 
\bea 
E + \lambda_{\min} < \frac{\pi}{2}, \label{eq:conditPhi}
\eea  the following inequality is satisfied:
$$\frac{\left|\Phi(p) - \BarPhi\right|}{\BarPhi} \le C_\varepsilon\,U_\varepsilon(E,\lambda_{\min}),$$
where
\bea 
U_\varepsilon(E,\lambda_{\min}) &=& \max\left(U_{1,\varepsilon},U_{2,\varepsilon},U_{3,\varepsilon}
\right)\nonumber\\
U_{1,\varepsilon}(E,\lambda_{\min}) &=& \frac{2\pi^\varepsilon}{3}\lambda_{\rm min}^{2-\varepsilon}, \nonumber\\
U_{2,\varepsilon}(E,\lambda_{\min}) &=& \frac{4E}{\pi}\left(1+2E\lambda_{\min}^{-\frac{1 + \varepsilon}{2}}\right)^2,\nonumber \\
U_{3,\varepsilon}(E,\lambda_{\min}) &=&  \frac{4E}{\pi\,\cos^{1+\varepsilon}E} \label{eq:U}
\eea
and the $\varepsilon$--dependent constant $C_\varepsilon$ is
$$C_\varepsilon = \frac{3\pi}{16}\int_0^\pi \left|\left(\frac{1}{\sin\frac{\lambda}{2}}\right)'\right|\,V^{1+\varepsilon}(\lambda) \dd\lambda$$
with
$$V(\lambda)=\left\{\begin{array}{ll} \sin\lambda & {\rm for}\quad 0 \le \lambda \le \frac{\pi}{2} \\
 1 & {\rm for }\quad \frac{\pi}{2} < \lambda \le \pi\end{array}\right. .$$
\label{th:Phi} 
\end{theorem}
Note that because of the singularity of $\Psi(x)$ the estimate depends not only on $E$, but also on the distance from the nearest black hole $\lambda_{\min}$. For sufficiently small $E$ we can neglect the cos function in the denominator of $U_{3,\varepsilon}$ and we see that the dependence is via a positive power in the first variable and via a negative one in the second one. Since we expect that both  go to 0 as $N \to \infty$ the convergence of $U_\varepsilon$ to 0 depends on the  convergence rates of both $E$ and $\lambda_{\min}$ to 0. We will discuss this issue in detail in Section \ref{sec:sequence}.

\begin{theorem}
Let \begin{itemize}
\item $\delta_{\min} = \min_{i=1\dots N}\,\min_{j=1\dots N,\,j\neq i}\,\Lambda(x_i,x_j)$,
\item $\delta_{\max} = \max_{i=1\dots N}\,\min_{j=1\dots N,\,j\neq i}\, \Lambda(x_i,x_j)$,
\item $\alpha_{\max} = \max_{i=1\dots N} \alpha_i$ 
\end{itemize}
and let $E$ be the modified discrepancy of $A$. Let $\varepsilon$ satisfy $0 < \varepsilon \le 1$. Then, assuming that 
\bea 
E + \delta_{\min}< \frac{\pi}{2}, \label{eq:conditM}
\eea the following inequality holds:
$$\frac{\left|\eeff M - M_\tot\right|}{ M_\tot} \le C_\varepsilon\,W_\varepsilon\left(E,\frac{\alpha_{\max}}{\alpha},\delta_{\max},\delta_{\min}\right) + \frac{\alpha_{\max}}{\alpha},$$
where
\bea 
W_\varepsilon\left(E,\gamma,\delta_{\max},\delta_{\min}\right) &=&  \max\left(W_{1,\varepsilon},W_{2,\varepsilon},W_{3,\varepsilon}\right), \nonumber\\
W_{1,\varepsilon}\left(E,\gamma,\delta_{\max},\delta_{\min}\right) &=& \frac{2\pi^\varepsilon}{3}\delta_{\max}^{2-\varepsilon}, \nonumber\\
W_{2,\varepsilon}\left(E,\gamma,\delta_{\max},\delta_{\min}\right) &=& \frac{4 
E}{\pi}\left(1 + E\,\delta_{\min}^{-\frac{1+\varepsilon}{2}}\right)^2 + 2\pi^{1+\varepsilon}\,\gamma\,\delta_{\min}^{-1-\varepsilon}, \nonumber\\
W_{3,\varepsilon}\left(E,\gamma,\delta_{\max},\delta_{\min}\right) &=&
\frac{\frac{4E}{\pi} + 2\gamma}{\cos^{1+\varepsilon}E} \label{eq:W}
\eea
and $C_\varepsilon$ is defined just like in Theorem \ref{th:Phi}.
 \label{th:M}
\end{theorem}
In this theorem the bounds depend on positive powers of $E$ as well as $\frac{\alpha_{\max}}{\alpha}$ and $\delta_{\max}$, but also
on the negative power of the distance between the two closest black holes $\delta_{\min}$. Again we expect all these quantities to vanish as we increase the number of black holes while keeping them evenly distributed, so $W_\epsilon$ goes to $0$ only if 
$\delta_{\min}$ does not go to zero too quickly. 


\section{Proof of Theorem \ref{th:Phi}}
\begin{proof}
Recall that the relative difference $\sigma_\Phi$ is given by (\ref{eq:r}). Introduce a spherical coordinate system $(\lambda, \vartheta, \varphi)$ in which $p$ lies at the pole $\lambda=0$. Since the Green's function $\Psi_p(x)$ depends only on $\lambda$, the integration over the whole sphere in $G(A,p)$ simplifies to a one--dimensional
integral. After simplification we obtain
\bea
G(A,p) =  \sum_{n=1}^N \frac{\alpha_n}{\alpha}\,\frac{1}{\sin \frac{\lambda_n}{2}} 
 -\frac{2}{\pi}\int_{0}^{\pi} \frac{\sin^2\lambda}{\sin\frac{\lambda}{2}}\,\dd\lambda. \label{eq:G1d}
\eea
In  \ref{app:Hlawka-Zaremba} we prove a version of the Hlawka--Zaremba identity \cite{Hlawka1961, Hlawka1961p, Zaremba1968} relating the expression above to the integral of the  discrepancy function, with a non--standard measure and a singularity at the origin:
\begin{lemma} \label{lm:Hlawka-Zaremba}
Let $f(\lambda)$ be continuous on $(0,\pi]$, possibly with a single pole at $\lambda=0$ of order $\lambda^{-q}$, $q>-3$ and let $0<\lambda_n \le \pi$, $n=1,\dots,N$. Then
\bea
\frac{2}{\pi}\int_{0}^{\pi} f(\lambda) \,\sin^2\lambda\,\dd\lambda - \sum_{n=1}^N \frac{\alpha_n}{\alpha}\,f(\lambda_n) =
-\int_0^\pi f'(\lambda)\,\disc(\lambda)\,\dd\lambda  \nonumber
\eea
where
\bea
\disc(\lambda) = \frac{2}{\pi}\Xi(\lambda) - \sum_{n=1}^N \frac{\alpha_n}{\alpha}\,\theta(\lambda-\lambda_n).\nonumber
\eea
\end{lemma}

Applying Lemma \ref{lm:Hlawka-Zaremba} to (\ref{eq:G1d}) yields
\bea
G(A,p) = \int_0^\pi \left(\frac{1}{\sin\frac{\lambda}{2}}\right)'\,\disc_p(\lambda)\,\dd\lambda. \label{eq:Gap2}
\eea
In the standard Koksma--Hlawka inequality the expression above is bounded by the product of  $\sup|\disc_p(\lambda)|$ and $ \int_0^\pi \left|\left(\frac{1}{\sin\frac{\lambda}{2}}\right)'\right|\dd\lambda$. As we mentioned before, in our case the integral is infinite and this bound is useless. We therefore need to modify   (\ref{eq:Gap2}) in order to obtain a finite integral. We begin by
introducing an auxiliary function
\bea
V(\lambda) = \left\{\begin{array}{lr} \sin\lambda & \textrm{for} \quad 0 \le \lambda \le \frac{\pi}{2} \\
                          1 & \textrm{for} \quad  \frac{\pi}{2} < \lambda \le \pi \end{array}\right. \label{eq:Vdef}
\eea and rewriting (\ref{eq:Gap2}) as
\bea
G(A,p) = \int_0^\pi V(\lambda)^{1+\varepsilon}\left(\frac{1}{\sin\frac{\lambda}{2}}\right)'\,\frac{\disc_p(\lambda)}{V(\lambda)^{1+\varepsilon}}\,\dd\lambda.\nonumber
\eea
 $G(A,p)$ is thus bounded by
\bea
|G(A,p)|\le \sup \frac{|\disc_p(\lambda)|}{V(\lambda)^{1+\varepsilon}}  \int_0^\pi V(\lambda)^{1+\varepsilon}\left|\left(\frac{1}{\sin\frac{\lambda}{2}}\right)'\right|\,\dd\lambda \label{eq:Gineq}
\eea
with the integral now convergent. The convergence was achieved at the expense of more tight bounds on $\disc_p(\lambda)$ we 
need near $\lambda=0$, where the denominator under the supremum function vanishes. Fortunately the supremum can be bounded by a function of the modified discrepancy at $p$ and the distance from $p$ to the nearest puncture:
\begin{lemma} \label{lm:bounds1}
Let $p \in S^3 \backslash \left\{x_1,\dots,x_n\right\}$, let $\lambda_{\min} = \min_{n=1,\dots,N} \Lambda(p,x_n)$ and let $E_p$ be the modified discrepancy of $A$ at $p$. Let $\varepsilon$ satisfy $0 < \varepsilon \le 1$.
Let $V(\lambda)$ be given by (\ref{eq:Vdef}). 
 Then, assuming that $E_p + \lambda_{\min} < \frac{\pi}{2}$, we have
\bea
\frac{\left|\disc_p(\lambda)\right|}{V(\lambda)^{1+\varepsilon}} \le \max\left(\frac{2\pi^\varepsilon}{3}\lambda_{\rm min}^{2-\varepsilon},
\frac{4E_p}{\pi}\left(1+E_p\lambda_{\min}^{-\frac{1 + \varepsilon}{2}}\right)^2,\frac{4E_p}{\pi\,\cos^{1+\varepsilon}E_p}\right).\nonumber
\eea
\end{lemma}
\begin{proof}
We divide the interval $(0,\pi)$ into 3 parts and separately derive bounds on the left hand side of the inequality in three cases
\begin{itemize}
\item $0 \le \lambda \le \lambda_{\min}$:
In this interval $\disc_p(\lambda) = \frac{2}{\pi} \Xi(\lambda)$. We note that $\Xi(\lambda) \le \frac{\lambda^3}{3}$ and
$V(\lambda) \ge \frac{\lambda}{\pi}$ over the whole interval, thus
$$\textrm{lhs} \le \frac{2}{3}\pi^{\varepsilon}\lambda^{2-\varepsilon} \le \frac{2}{3}\pi^{\varepsilon}\lambda_{\min}^{2-\varepsilon}.$$ 
The last inequality holds because the expression is non--decreasing in $\lambda$.
\item $\lambda_{\min}\le \lambda \le \frac{\pi}{2} - E_p$: The numerator  is by assumption bounded by $\frac{4E_p}{\pi} \sin^2(\lambda + E_p)$ and $V(\lambda) = \sin\lambda$ in this interval, thus
\bea 
\fl\textrm{lhs} &\le & \frac{\frac{4E_p}{\pi}\sin^2(E_p + \lambda)}{\sin^{1+\varepsilon} \lambda} = \frac{\frac{4E_p}{\pi}\left(\sin\lambda \cos E_p + \cos\lambda \sin E_p\right)^2}{\sin^{1+\varepsilon} \lambda}   \nonumber\\
\fl&=& \frac{4E_p}{\pi} \left(\sin^{1/2 - \varepsilon/2}\lambda \cos E_p + \sin E_p \cos^{1/2 - \varepsilon/2}\lambda \cot^{1/2 + \varepsilon/2}\lambda\right)^2\nonumber. \eea
$\cot \lambda$ is bounded from above by $\frac{1}{\lambda}$ over the whole  interval considered, $\sin E_p$ by $E_p$, other trigonometric functions by 1. We get
\bea
\textrm{lhs} \le \frac{4E_p}{\pi} \left(1 + E_p\,\lambda^{-1/2 - \varepsilon/2}\right)^2 \le \frac{4E_p}{\pi} \left(1 + E_p\,\lambda_{\min}^{-1/2 - \varepsilon/2}\right)^2,\nonumber
\eea
the last inequality following from the fact that the expression is decreasing in $\lambda$. 
\item $\frac{\pi}{2} - E_p \le \lambda \le \pi$: It follows from the assumptions that $\left|\disc_p(\lambda)\right| \le \frac{4E_p}{\pi} \sin^2 (E_p 
+ \lambda) \le \frac{4E_p}{\pi}$. The denominator on the other hand is non--decreasing and therefore satisfies $V(\lambda) \ge V(\frac{\pi}{2} - E_p)$. Thus
$$\textrm{lhs} \le \frac{\frac{4E_p}{\pi}}{\sin^{1+\varepsilon}\left(\frac{\pi}{2}-E_p\right)} = \frac{4E_p}{\pi\,\cos^{1+\varepsilon}E_p}.$$
\end{itemize}
The three inequalities above, taken together, prove the lemma. \qed
\end{proof}
We come back to the main theorem. From Lemma \ref{lm:bounds1}  and (\ref{eq:Gineq}) it follows immediately that
\bea
|G(A,p)| \le U_\varepsilon(E_p, \lambda_{\min}) \int_0^\pi V(\lambda)^{1+\varepsilon}\left|\left(\frac{1}{\sin\frac{\lambda}{2}}\right)'\right|\,\dd\lambda.\label{eq:GUE1}\eea
By definition $E_p\le E$, and by assumption $E < \frac{\pi}{2}$. Since $U_\varepsilon$ is non--decreasing as a function of its first variable in the interval $(0, \frac{\pi}{2})$,   we can replace $E_p$ with $E$ in (\ref{eq:GUE1}). Substituting this inequality back to  (\ref{eq:r}) yields the result.\qed
\end{proof}

\section{Proof of Theorem \ref{th:M}}
\begin{proof}
The proof is quite similar to the previous one, with a few complications.
Recall that the relative difference  $\sigma_M$  is given by (\ref{eq:s}). $G(\widetilde A_n,x_n)$ after similar simplifications reads 
\bea
G(\widetilde A_n,x_n) =  \sum_{k=1,\, k\neq n}^N \frac{\alpha_k}{\alpha - \alpha_n}\,\frac{1}{\sin \frac{\lambda_n}{2}} 
 -\frac{2}{\pi}\int_{0}^{\pi} \frac{\sin^2\lambda}{\sin\frac{\lambda}{2}}\,\dd\lambda. \label{eq:Gnoxn1}
\eea
We apply the Hlawka--Zaremba inequality from \ref{app:Hlawka-Zaremba} and introduce the powers of function $V(\lambda)$ to get
\bea
G(\widetilde A_n,x_n) = \int_0^\pi V(\lambda)^{1+\varepsilon}\left(\frac{1}{\sin\frac{\lambda}{2}}\right)'\,\frac{\widetilde\disc_{x_n}(\lambda)}{V(\lambda)^{1+\varepsilon}}\,\dd\lambda. \label{eq:Gnoxn2}
\eea
 $\widetilde\disc_{x_n}(\lambda)$ denotes here the discrepancy function for the configuration $\widetilde A_n = A\backslash \{(x_m,\alpha_n)\}$. We would like to relate it to the discrepancy function of the whole $A$:
\bea
\widetilde\disc_{x_n}(\lambda) - \disc_{x_n}(\lambda) = \frac{-\alpha_n}{\alpha(\alpha - \alpha_n)} \sum_{k=1,\, k\neq n}^N  \alpha_k \theta(\lambda - \lambda_k) + \frac{\alpha_n}{\alpha}\theta(\lambda - \lambda_n) \nonumber.
\eea
It is easy to turn the equation above into the bound for the difference between both functions:
\bea
|\widetilde\disc_{x_n}(\lambda) - \disc_{x_n}(\lambda)| \le \left|\frac{-\alpha_n}{\alpha}\right|  + \frac{\alpha_n}{\alpha} = \frac{2\alpha_n}{\alpha} \nonumber
\eea
or
\bea
|\widetilde\disc_{x_n}(\lambda)|  \le  \frac{2\alpha_n}{\alpha} +| \disc_{x_n}(\lambda)|. \label{eq:tildedisc}
\eea

We are now ready to prove a modified version of Lemma \ref{lm:bounds1}:
\begin{lemma} \label{lm:bounds2}
Let $\widetilde \disc_p(\lambda)$ denote the discrepancy function for $A\backslash \left\{\left(x_n,\alpha_n\right)\right\}$. Let  $\delta_{n} =  \min_{m=1,\dots,N, m\neq n} \Lambda(x_m,x_n)$ and let $E$ be the modified discrepancy of $A$. Let $\varepsilon$ satisfy $0 < \varepsilon \le 1$.
Let $V(\lambda)$ denote the function 
\bea
V(\lambda) = \left\{\begin{array}{lr} \sin\lambda & \textrm{for} \quad 0 \le \lambda \le \frac{\pi}{2} \\
                          1 & \textrm{for} \quad  \frac{\pi}{2} < \lambda \le \pi \end{array}\right. .\nonumber
\eea
 Then, assuming that $E + \delta_{n} \le \frac{\pi}{2}$, we have
\bea
\fl\frac{\left|\widetilde\disc_{x_n}(\lambda)\right|}{V(\lambda)^{1+\varepsilon}} \le \nonumber 
\max&&\left(\frac{2\pi^\varepsilon}{3}\delta_n^{2-\varepsilon},
\frac{4E}{\pi}\left(1 + E\,\delta_n^{-\frac{1+\varepsilon}{2}}\right)^2 + 2\pi^{1+\varepsilon}\,\frac{\alpha_{\max}}{\alpha}\,\delta_n^{-1-\varepsilon},\right.\\
&& \left.
\frac{\frac{4E}{\pi} + 2\alpha_{\max}/\alpha}{\cos^{1+\varepsilon}E}\right)\nonumber
\eea
\end{lemma}

\begin{proof}
We divide the interval $(0,\pi)$ into 3 parts and separately derive bounds on the left hand side of the inequality in three cases
\begin{itemize}
\item $0 \le \lambda \le \delta_n$:
We repeat the reasoning from the proof of Lemma \ref{lm:bounds1}. In this interval $\widetilde\disc_p(\lambda) = \frac{2}{\pi} \Xi(\lambda)$, so 
$$
\textrm{lhs} \le   \frac{2}{3}\pi^{\varepsilon}\delta_n^{2-\varepsilon}.$$ 
\item $\delta_{n}\le \lambda \le \frac{\pi}{2} - E_{x_n}$: From (\ref{eq:tildedisc}) we know that 
\bea
\textrm{lhs} \le \left|\frac{\disc_{x_n}(\lambda)}{V(\lambda)^{1+\varepsilon}}\right| + \frac{2\alpha_n / \alpha}{V(\lambda)^{1+\varepsilon}}.\label{eq:splitbounds}
\eea
The first term can be bounded just like in the proof of Lemma \ref{lm:bounds1}, the second by $\frac{2\alpha_n / \alpha}{V(\delta_n)^{1+\varepsilon}}$. Using the fact that $V(\lambda) \ge \frac{\lambda}{\pi}$ we can bound this expression by  $\frac{2\alpha_n / \alpha}{\left(\delta_n/\pi\right)^{1+\varepsilon}}$. Taken together, these estimates prove that
\bea
\textrm{lhs} \le \frac{4E_{x_n}}{\pi}\left(1 + E_{x_n} \delta_n^{-\frac{1+\varepsilon}{2}}\right)^2 + 2\pi^{1+\varepsilon}\frac{\alpha_n}{ \alpha}\delta_n^{-1-\varepsilon}.\nonumber
\eea

\item $\frac{\pi}{2} - E_{x_n} \le \lambda \le \pi$: Again we split the left hand side according to (\ref{eq:splitbounds})
and use the proof of Lemma \ref{lm:bounds1} to bound the first term. The second one on the other hand satisfies
$$
\frac{2\alpha_n/\alpha}{V(\lambda)^{1+\varepsilon}} \le  \frac{2\alpha_n/\alpha}{\sin^{1+\varepsilon}\left(\frac{\pi}{2} - E_{x_n}\right)} = \frac{2\alpha_n/\alpha}{\cos^{1+\varepsilon}\ E_{x_n}},
$$
the inequality due to the fact that $V(\lambda)$ is non--decreasing. It follows that
$$
\frac{2\alpha_n/\alpha}{V(\lambda)^{1+\varepsilon}} \le \frac{\frac{4E_{x_n}}{\pi} + 2\alpha_{\max}/\alpha}{\cos^{1+\varepsilon}E_{x_n}}.
$$
\end{itemize}
The three inequalities above, taken together, prove the lemma. \qed
\end{proof}

With Lemma \ref{lm:bounds2} established equation (\ref{eq:Gnoxn2}) yields the following estimate:
\bea
\fl|G(\widetilde A_n,x_n)| &\le&  \sup  \frac{|\widetilde\disc_{x_n}(\lambda)|}{V(\lambda)^{1+\varepsilon}}  \int_0^\pi V(\lambda)^{1+\varepsilon}\left|\left(\frac{1}{\sin\frac{\lambda}{2}}\right)'\right|\,\dd\lambda \nonumber \\
\fl&\le & W_\varepsilon\left(E_{x_n},\frac{\alpha_n}{\alpha},\delta_n,\delta_n\right)\,\int_0^\pi V(\lambda)^{1+\varepsilon}\left|\left(\frac{1}{\sin\frac{\lambda}{2}}\right)'\right|\,\dd\lambda \nonumber
\eea
By definition $E_{x_n}\le E$, and by assumption $E < \frac{\pi}{2}$. We note that $W_\varepsilon$ is non--decreasing as a function of its first three variables in the  and non--increasing in the fourth one as long as the first one is smaller than $\frac{\pi}{2}$. This condition is satisfied by assumption and therefore  the maximum over all 
punctures, appearing in (\ref{eq:s}), can be estimated as well:
\bea
\fl\max_{n=1\dots N} |G(\widetilde A_n,x_n)| \le W_\varepsilon\left(E,\frac{\alpha_{\max}}{\alpha},\delta_{\min},\delta_{\max}\right)\,\int_0^\pi V(\lambda)^{1+\varepsilon}\left|\left(\frac{1}{\sin\frac{\lambda}{2}}\right)'\right|\,\dd\lambda \nonumber.
\eea
Combining this inequality with (\ref{eq:s}) completes the proof. \qed
\end{proof}

\section{Configurations with a continuum limit}\label{sec:sequence}

In this section we will present an infinite sequence of configurations in which an increasing number of black holes covers uniformly
$S^3$. The main challenge is to make sure that the distribution is sufficiently uniform so that the global modified  discrepancy converges to 0 sufficiently fast.  We will later apply  Theorems \ref{th:Phi} and \ref{th:M} to establish the existence of the continuum limit and 
the vanishing of backreaction for this sequence. 
We begin by introducing a special class of configurations constructed from partitions of $S^3$. By a partition of $S^3$ we understand a
finite set $\{P_i\}$ of subsets of $S^3$ such that
\begin{enumerate}
\item $\bigcup_i P_i = S^3$,
\item each $P_i$ is closed, simply connected and has non--vanishing volume with respect to the "round" volume form,
\item for each pair $P_i$, $P_j$, $i\neq j$, their intersection is either empty or contained within the boundaries of both $P_i$ and $P_j$.
\end{enumerate}

The construction of $A$ proceeds as follows: for a fixed partition we pick up the punctures $x_i$ such that every $x_i \in P_i$. We then fix the mass parameters $\alpha_i$
by demanding that they are proportional to the volume of its corresponding domain $P_i$, i.e.
\bea
\alpha_i = \frac{\alpha}{2\pi^2}\vol (P_i) \label{eq:alphai1}
\eea 
with arbitrary, positive $\alpha$.

Note that all 6 regular configurations can be obtained from a construction of this type. Namely, we need to  take $P_i$ to be one of 6 regular tessellations of $S^3$ and put the punctures at the 
geometric center of each domain. 

The key feature of the configurations constructed from partitions is that we may bound their global modified discrepancy in a very simple way:
\begin{lemma} \label{lm:Ediam}
Let $A$ be a configuration constructed from a partition $\{P_i\}$ as described above. Then the global modified discrepancy satisfies the following bound:
\bea
E \le \diam\left(\{P_i\}\right),\nonumber
\eea
where $\diam$ is the diameter of the partition, i.e.
\bea
\diam\left(\left\{P_i\right\}\right) = \max_{i=1..N} \sup_{a,b \in P_i} \Lambda(a,b). \label{eq:diamdef}
\eea
\end{lemma}
\begin{proof}

Let $p$ be a point in $S^3$ and let $\lambda$ be a number satisfying $0 \le \lambda \le \pi$. Consider the cap $B(p,\lambda)$. The domains $P_i$ can be divided into
three classes: those which lie entirely within $B(p,\lambda)$, those which lie entirely outside, and the rest.  The set of their indices will be denoted by $I_{\rm in}$, $I_{\rm out}$ and $I_0$ respectively (see figure \ref{fig:domains}). Note that the domains in the third class must have a non--vanishing intersection with the boundary of  $B(p,\lambda)$.

\bfi
\bce
\includegraphics[width=0.80\textwidth]{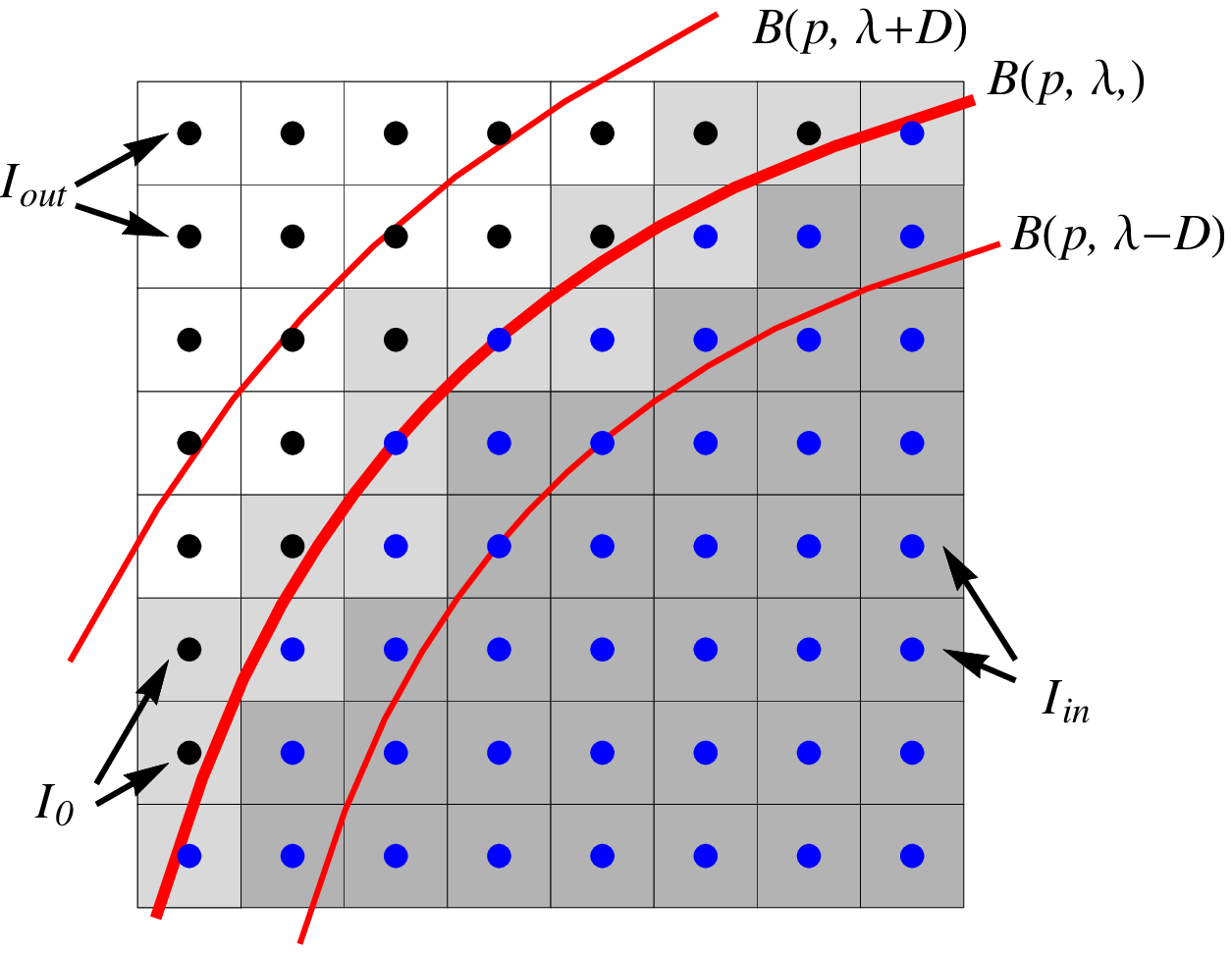}
\caption{Domains within $I_\textrm{in}$, $I_\textrm{out}$ and $I_0$ (shades of grey). Black dots denote the punctures lying outside $B(p,\lambda)$, while the blue one denote those which lie within. Moreover, all domains in $I_0$ lie within the hollow cap between $B(p,\lambda - D)$ and $\tilde B(p,\lambda - D)$}
\label{fig:domains}
\ece
\efi
Consider now the sum of the mass parameters of all punctures lying within the cap $B$.  It can be decomposed into the sum over $I_{\rm in}$ and over a part of $I_0$ (see figure \ref{fig:domains}):
\bea
\sum_{x_i\in B(p,\lambda)} \alpha_i = \sum_{i\in \Iin} \alpha_i + \sum_{i\in I_0,\,x_i\in B(p,\lambda)} \alpha_i \label{eq:summ1}.
\eea
On the other hand
\bea
 \sum_{i\in I_0, \, x_i\in B(p,\lambda)} \alpha_i \le \sum_{i\in I_0} \alpha_i \label{eq:summ2}.
\eea
Taking (\ref{eq:summ1}) and (\ref{eq:summ2}) together we obtain the following bound:
\bea
\sum_{i\in \Iin} \alpha_i \le \sum_{x_i\in B(p,\lambda)} \alpha_i  \le \sum_{i\in \Iin} \alpha_i + \sum_{i\in I_0} \alpha_i . \label{eq:summ3}
\eea

The lower and upper bounds are by assumption equal to $\frac{\alpha}{2\pi^2}\sum_{i\in \Iin} \vol(P_i)$ and $\frac{\alpha}{2\pi^2}\left(\sum_{i\in \Iin} \vol(P_i) + \sum_{i\in I_0} \vol(P_i)\right)$ 
respectively (see equation (\ref{eq:alphai1})). Note that all points of the union $\bigcup_{i\in I_0} P_i$ lie not further than $D = \diam(\{P_i\})$ from the boundary of $B(p,\lambda)$ and therefore 
they are all contained within the (hollow) cap given by $C=\{x\in S^3\left|\right.\lambda - D \le \Lambda(x,p) \le \lambda + D\}$ (see again figure \ref{fig:domains}). Two inclusions follow easily: $\tilde B(p,\lambda - D) \subset \bigcup_{i\in \Iin} P_i$ and $ \bigcup_{i\in \Iin} P_i \cup \bigcup_{i \in I_0} P_i \subset B(p,\lambda + D)$, where $\tilde B(p,\xi)$ denotes the cap $B(p,\xi)$ without its boundary, i.e. $\tilde B(p,\xi) = \{x\in S^3\left|\right. \Lambda(p,x) < \xi\}$. Thus $\vol(B(p,\lambda - D)) \le \sum_{i\in \Iin} \vol(P_i)$ and $\sum_{i\in \Iin} \vol(P_i) + \sum_{i\in I_0} \vol(P_i) \le \vol(B(p,\lambda + D))$. Combining this with (\ref{eq:summ3}) we obtain
\bea
\frac{\alpha}{2\pi^2}\vol(B(p,\lambda-D)) \le \sum_{x_i\in B(p,\lambda)} \alpha_i  \le \frac{\alpha}{2\pi^2}\vol(B(p,\lambda+D))  \label{eq:summ4}
\eea
or 
\bea
\vol(B(p,\lambda-D)) \le \frac{2\pi^2}{\alpha}\sum_{x_i\in B(p,\lambda)} \alpha_i  \le \vol(B(p,\lambda+D))  \label{eq:summ4b}
\eea
We will now use the bounds above to estimate the discrepancy function. Since the volume of a cap is non--decreasing with radius we have
\bea
\vol(B(p,\lambda-D)) \le \vol(B(p,\lambda))  \le \vol(B(p,\lambda+D)) . \nonumber
\eea
Thus both $\vol(B(p,\lambda))$ and $\frac{2\pi^2}{\alpha}\sum_{x_i\in B(p,\lambda)} \alpha_i$ lie inside the interval with endpoints at $\vol(B(p,\lambda-D))$ and 
$\vol(B(p,\lambda+D))$, so the modulus of their difference cannot exceed the length of the interval:
\bea
\left|\frac{2\pi^2}{\alpha}\sum_{x_i\in B(p,\lambda)} \alpha_i-\vol(B(p,\lambda))\right| \le \vol(B(p,\lambda+D))-\vol(B(p,\lambda-D)), \nonumber
\eea
or, after dividing by $2\pi^2$,
\bea
\left|\disc_p(\lambda)\right| \le \frac{2}{\pi}\left(\Xi(\lambda+D)-\Xi(\lambda-D)\right).\nonumber
\eea
Note that in the inequality above we have extended function $\Xi(x)$ to a continuously differentiable function over the whole $\Rr$ by assuming that it vanishes  for $x < 0$ and is equal to $\frac{\pi}{2}$ for $x > \pi$. 
$\Xi(x)$ is convex for $x \le \frac{\pi}{2}$, so if $\lambda \le \frac{\pi}{2} - D$, then the right hand side can be bounded by $\frac{2}{\pi}(2D\cdot \Xi'(\lambda + D)) = \frac{4D}{\pi}\sin^2(\lambda + D)$. On the other hand, if $\lambda > \frac{\pi}{2} - D$, then the right hand side is bounded by $\frac{2}{\pi}(2D\cdot \sup \Xi'(x))$. Since $\Xi'(x)  = \sin^2 x $ or $0$, the supremum is equal to 1 and $|\disc_p(\lambda)|$ is bounded by $\frac{4D}{\pi}$. Thus 
\bea
\left|\disc_p(\lambda)\right| \le F_D(\lambda)\nonumber
\eea
and, what follows easily, $E_p \le D$. Since this holds for all points $p$, we have $E \le D$.
This completes the proof. \qed 
\end{proof} 

\subsection{Construction of the sequence}

We are ready to discuss the construction of the infinite sequence $H_N$ of configurations which approach the FLRW continuum limit as $N \to \infty$. 
We begin repeating partially the construction leading to the universe with a 8-black-hole regular lattice \cite{Bentivegna:2012ei}, i.e. dividing $S^3$ into eight identical, quasi-cubic domains. Let $S^3$ be embedded in  $\Rr^4$ as a
sphere of unit radius. We take eight points  $(\pm 1, 0, 0, 0)$, $(0, \pm 1 , 0,0)$, $(0,0,\pm 1,0)$, $(0,0,0,\pm 1)$, denoted by $r_i$, 
corresponding to the vertices of an inscribed regular polytope, i.e. a 16--cell. The eight domains $\Sigma_i$, $i=1\dots 8$, are defined
as the Voronoi cells of $r_i$'s, i.e. $\Sigma_k = \{x \in S^3 | \Lambda(x,r_k) \le \Lambda(x,r_i)\quad \textrm{for all }i\neq k\}$.
The cells are quasi--cubes with 6 faces, 12 edges and 8 vertices \cite{Bentivegna:2012ei}. The partition we obtain is quite similar to the  tessellation of $\Rr^3$ with cubes, however note that only 4 instead of 6 cell edges meet at each vertex. 

We will now subdivide each of the cells into subregions. We first introduce a coordinate system on each of the cells using a 
projection centered at the origin onto a three--dimensional subspace. Consider the cell around the point $(1,0,0,0)$ and define the mapping
\bea
\psi: S^3 \ni (x_1,x_2,x_3,x_4) \mapsto (x_2/x_1,x_3/x_1,x_4/x_1) \in \Rr^3 \label{eq:Psiproj}
\eea
The projection does not conserve  angles or distances, but it maps great circles and great spheres into straight lines and planes respectively. 
In particular, it maps the cell in question into a true cube given by $-1 \le z_1 \le 1$, $-1 \le z_2 \le 1$, $-1 \le z_3 \le 1$, $z_a$ being the Cartesian coordinates on the target $\Rr^3$. The resulting cube in $\Rr^3$ can be subdivided into $n^3$ identical cubes $\Pi_i$ in a straightforward manner. This partition can now be pulled back to $S^3$ as a partition of the big quasi--cube into $n^3$ non--identical small  quasi--cubes whose faces are made of big spheres (see figure \ref{fig:projection}). Let $q_i$ denote the centers of each cube. We will place the punctures on $S^3$ at $p_i = \psi^{-1}(q_i)$  (see again figure \ref{fig:projection}). The construction
can be repeated for all eight quasi--cubes, yielding ultimately a partition of $S^3$ into $N=8n^3$ domains $P_i$ and punctures $p_i$
within the corresponding domains. Following the prescription from the previous subsection we  take $\alpha_i = \vol(P_i) \frac{\alpha}{2\pi^2}$, with arbitrary positive $\alpha$ setting the scale. 
\bfi
\bce
\includegraphics[width=0.7\textwidth]{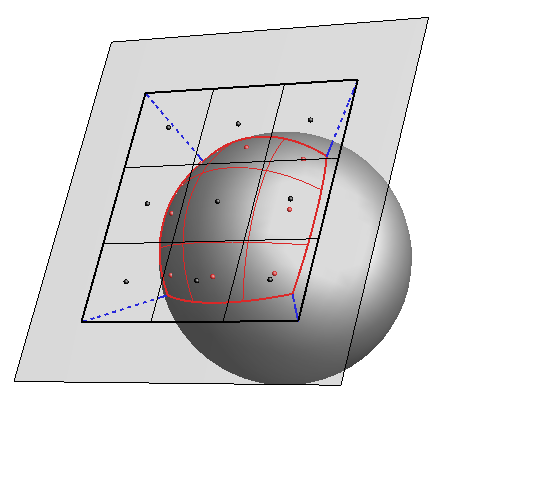}
\caption{Construction of the partition of a big quasi--cube in $S^3$ and the location of the punctures (red) via the projection  of the partition of the corresponding cube by $\psi$.  One dimension suppressed.}
\label{fig:projection}
\ece
\efi

The configurations $H_N$ we constructed above, each containing $N=8n^3$ black holes, form an infinite sequence whose  first elements are $H_8$, $H_{64}$, $H_{216}$, $H_{512}$, $H_{1000}$ and $H_{1728}$. $H_8$ is identical to the original regular 8-black-hole model $R_8$ and
all subsequent configurations inherit the whole group of discrete symmetries of  $R_8$. In particular, it is noteworthy that the edges of the big quasi--cubes in all $H_N$  are locally rotationally symmetric curves in the sense of \cite{Clifton:2013jpa} and the exact, relativistic evolution along them can be investigated using the same methods.

\subsection{Estimates of the parameters of the constructed lattice}
We will now explain in what sense the configurations $H_N$ approach the continuum limit as $N\to\infty$. In order to make use of Theorems \ref{th:Phi} and \ref{th:M} we need to estimate their diameter of the partition, as well as $\lambda_{\min}$, $\delta_{\min}$, $\alpha_{\max}$ and $\delta_{\max}$. 
Obviously the geometry of the quasi--cubic lattice on $S^3$ is significantly more complicated than the geometry of the corresponding cubic lattice on $\Rr^3$. It is therefore a reasonable strategy to relate all parameters of $\{P_i\}$ to the corresponding parameters of
 $\{\Pi_i\}$. 
 
 We will  proceed in the following way: first, we will estimate $\delta_{\min}$ and $\delta_{\max}$ by certain maxima or minima over all \emph{single} small quasi--cubes of the partition of $S^3$. This is necessary because unlike $\diam(\{P_i\})$ or $\mu_{\min}$ in these two we take supremum or infimum over \emph{pairs} of small quasi--cubes.  Note that the partition of $S^3$ consists of 8 identical but distinct partitions of big quasi--cubes and if the parameter in question refers to a pair of domains, then we would have to consider separately what happens if the small quasi--cubes belong to two different large quasi--cubes. This is not needed if we minimize or maximize over individual cells because each minimum or maximum over a single Voronoi cell is equal to the maximum or minimum over the whole $S^3$.
 
 In the second step we will apply  the following inequality, proved in  \ref{app:distances}, estimating the distance $\Lambda(\cdot,\cdot)$ by
the Euclidean distance $D(\cdot,\cdot)$ on $\Rr^3$: for all $x,y$ inside the same single big quasi--cube  we have
\bea
\frac{1}{4} D(\Psi(x),\Psi(y)) \le \Lambda(x,y) \le  D(\Psi(x),\Psi(y)), \label{eq:DLD}
\eea
where $\Psi$ is the corresponding projection of the big quasi--cube. 
We will use (\ref{eq:DLD}) to estimate the parameters by analogous parameters of the cubic partition in $\Rr^3$. These in turn  are
straightforward to  read out thanks to their very simple geometry.

Recall the definition of the diameter of $\{P_i\}$ (\ref{eq:diamdef}). 
From (\ref{eq:DLD}) we have easily
\bea
\Lambda(x,y) \le D(\Psi(x),\Psi(y))\nonumber
\eea
for all $x$ and $y$ inside the same big quasi--cube. Moreover, if $x,y \in P_i$, then $\Psi(x),\Psi(y) \in \Pi_i$. Thus the inequality carries over to the supremum
taken over all $x,y$ in $P_i$ and over all $P_i$'s. Therefore
\bea
\diam(\{P_i\}) \le \diam(\{\Pi_i\}),\nonumber
\eea
where the second diameter is calculated with respect to the flat metric.  The latter diameter is equal to the length of the space diagonal of a single cube, equal to $\frac{2\sqrt{3}}{n}$. Thus
\bea
E \le \diam(\{P_{i}\}) \le \frac{2\sqrt{3}}{n} = 4\sqrt{3}{N^{-1/3}},\nonumber
\eea
we have reexpressed here the bound by the number of punctures.

Other parameters can be estimated in a similar manner. Let $P_i$ be a domain and $P_j$ a  neighboring domain, and let $u$ be a point in the boundary of both domains, or $u \in \partial P_i \cap \partial P_j$. The triangle inequality yields
\bea
\Lambda(x_i,x_j) \le \Lambda(x_i,u)  +  \Lambda(u,x_j) .\nonumber
\eea
Both terms in the last expression are bounded from above by the partition diameter. Thus
\bea
\delta_{\max} &=& \max_i \min_{j\neq i} \Lambda(x_i,x_j)  \le 2\,\diam(\{P_i\})  \nonumber\\ &\le&  2\,\diam(\{\Pi_i\}) = 8\sqrt{3} N^{-1/3},\nonumber
\eea
the last two inequalities have been proved above.

Let now $P_i$ and $P_j$ be two distinct small quasi--cubical domains. Let $r$ and $s$ denote
the intersections of the shortest geodesic joining $x_i$ and $x_j$ with $\partial P_i$ and $\partial P_j$ respectively ($r$ might be the same as $s$, but does not have to). Then 
\bea
\Lambda(x_i,x_j) = \Lambda(x_i,r)  + \Lambda(r,s) + \Lambda(s,x_j) \ge  \Lambda(x_i,r) + \Lambda(s,x_j).\nonumber
\eea
Both latter terms are bounded by the smallest distance between $x_i$ and the boundary of its domain:
\bea
\Lambda(x_i,r) \ge  \min_k \inf_{s\in \partial P_k}  \Lambda(s,x_k) \ge \frac{1}{4} \min_k \inf_{t\in \partial \Pi_k} D(t,q_k),\nonumber
\eea
the last inequality follows from (\ref{eq:DLD}). The last term is simply the distance between the center of a small cube and its face,
i.e. $\frac{1}{n}$, thus
\bea
\delta_{\min} = \min_i \min_{j\neq i} \Lambda(x_i,x_j) \ge\frac{2}{4n}  =  N^{-1/3}.\nonumber
\eea

Finally consider a point $p$ lying on one of the domain boundaries. Then 
\bea
\lambda_{\min} = \min_i \Lambda(p,x_i) \ge \min_i \inf_{r\in \partial P_i} \Lambda(r,x_i)  .\nonumber
\eea
Again we make use of (\ref{eq:DLD}):
\bea
 \min_i \inf_{r\in \partial P_i} \Lambda(r,x_i) \ge \frac{1}{4}  \min_i \inf_{s\in \partial \Pi_i} D(s,q_i) = \frac{1}{2}N^{-1/3}  .\nonumber
\eea
Last equality was proved above.

Let us finally have a look at $\alpha_i$. In  \ref{app:distances} we prove an estimate for the volume of a set measured by $\eta$ and the corresponding volume
on the target flat space of the projection $\Psi$. From (\ref{eq:volumes}) we obtain
$\vol(P_i) \le \vol_{\Rr^3}(\Pi_i) = 8 n^{-3} = 64 N^{-1}$  
and hence $2\pi^2\frac{\alpha_i}{\alpha} \le  64 N^{-1}$, so
\bea
\alpha_{\max} \le \frac{64\alpha}{2\pi^2} N^{-1}.\nonumber
\eea 

\subsection {The continuum limit and asymptotics}
We are ready to apply Theorems \ref{th:Phi}  and  \ref{th:M}  to configurations $H_N$ in order to understand their properties  for large $N$. 

The distance from any point $p$ and the nearest black hole $\lambda_{\min}$ cannot exceed the diameter of the partition which scales like 
$O(N^{-1/3})$. We may of course consider points which approach a black hole faster than that, i.e. like $O(N^{-\alpha})$ with any
$\alpha \ge \frac{1}{3}$. 

Fix $\alpha$ such that $\frac{1}{3} < \alpha < 1$ and a positive constant $C$. Let $J^\alpha_{C,N}$ denote the set of points which lie further that $C N^{-\alpha}$
from the nearest black hole, i.e. $J^\alpha_{C,N} = \left\{p \in S^3 | \min_n \Lambda(x_n,p) \ge C \cdot N^{-\alpha}\right\}$.
If we now plug in $E=O(N^{-1/3}$ and $\lambda_{\min} = O(N^{-\alpha})$ into Theorem \ref{th:Phi}, we see that  for the points in $J^\alpha_{C,N}$ the relative deviation of $\Phi$ from the mean satisfies $\sigma_\Phi \le \max(O(N^{-1/3}),O(N^{\alpha(1+\varepsilon) - 1}))$. We still have the freedom to choose $\varepsilon$. Note that the second exponent is negative if we choose it to be small enough, so we conclude that for points in $J^\alpha_{C,N}$ the conformal factor, and thus also the metric tensor itself, tends to the metric tensor 
of the corresponding FLRW model. 

Consider now the volume occupied by $J^\alpha_{C,N}$ and measured with respect to the $S^3$ volume form. 
Recall that the total volume of $S^3$ is equal to $2\pi^2$. Obviously the volume of the compliment of $J^\alpha_{C,N}$ is bounded from above by the sum of volumes $N$
balls centered at $x_n$ of radius $C N^{-\alpha}$, so the volume of $J^\alpha_{C,N}$ itself satisfies
\bea
\vol(J^\alpha_{C,N}) \le 2\pi^2 - \sum_{n=1}^N \vol(B(x_n,C N^{-\alpha}) = 2\pi^2 - O(N^{-3\alpha + 1}).\nonumber
\eea
This expression goes to $2\pi^2$ for $\alpha > \frac{1}{3}$. This means that for large $N$ the set $J^\alpha_{C,N}$, in
which the value of $\Phi$ approaches the mean as $N\to\infty$, takes up almost the whole original sphere $S^3$. In fact, the only places where 
we cannot prove the convergence to $\Mean{\Phi}$ are points lying asymptotically within the distance $O(N^{-1})$ from the nearest puncture.

Let us now focus at the immediate vicinity of a puncture, where Theorem \ref{th:Phi} fails to guarantee the convergence of the conformal factor to the mean value. We introduce a spherical coordinate system $(\lambda,\vartheta,\varphi)$ centered at a puncture $x_n$ and a new, rescaled  radial coordinate $\tilde\lambda = \frac{\alpha}{\alpha_n}\lambda$.
$\tilde\lambda$ is a zooming--in coordinate, because for a fixed $\tilde\lambda$ the corresponding $\lambda$
scales like $\alpha_n/\alpha = O(N^{-1})$. Now the conformal factor takes the form of
\bea
\Phi(x) &=& \frac{16\alpha}{3\pi} + \frac{\alpha_n}{\sin \frac{\tilde\lambda \alpha_n}{2 \alpha}} -\frac{16\alpha_n}{3\pi}
+ (\alpha - \alpha_n) G(\widetilde A_n,x) = \nonumber\\
&& = \alpha\left(\frac{16\alpha}{3\pi} +\frac{2}{\tilde\lambda} + \textrm{h.o.t.} \right). \label{eq:PhiZoom}
\eea
We have neglected here the terms proportional to $\frac{\alpha_n}{\alpha}$ and $G(\widetilde A_n,x)$ as they vanish
for large $N$ and expanded the  sin function in the denominator of the second term.
Note that the metric $q_0$ converges in the new coordinates to a flat one as $N\to\infty$:
\bea
q_0 = \left(\frac{\alpha_n}{\alpha}\right)^2\left(\dd\tilde\lambda^2 + \tilde\lambda^2 \dd\Omega^2\right) + \textrm{h.o.t.},\nonumber
\label{eq:q0Zoom}
\eea
again we have only kept the leading order terms in the expansion in powers of $N$.
Metric (\ref{eq:q0Zoom}), together with the  conformal factor (\ref{eq:PhiZoom}), corresponds to  a single Schwarzschild black 
hole in rescaled conformal coordinates. The horizon of this black hole lies at $\tilde\lambda = \frac{3\pi}{8}$, its ADM mass
is equal to $\frac{64}{3\pi}\alpha\alpha_n$ and its Schwarzschild radius, expressed in terms of $\lambda$, is equal
to $\lambda = \frac{128}{3\pi}\alpha\alpha_n$. 

The picture which emerges from these consideration can be summarized as follows: 
the metric in the region within the distance of $O(N^{-1})$ from the puncture indeed does not converge to the FLRW metric. This is because  asymptotically it lies within a finite number of Schwarzschild radii from the nearest black hole and its metric is strongly perturbed by  
its presence. This distortion persists for arbitrary large $N$. In fact, asymptotically the spatial metric tensor  takes the form of a single Schwarzschild black hole relatively isolated from the rest of the black holes:   recall that the distance from the nearest other BH is bounded from below by $\delta_{\min} = O(N^{-1/3})$. It is only in the region between the black holes  that the spatial metric tends to the homogeneous one. On the other hand, the volume of this faraway region grows with $N$ and asymptotically it takes
up almost all of the original $S^3$.

These results  may seem quite puzzling at first. One might naively expect that as we add more and more black holes into a compact manifold like $S^3$ the Universe will become more and more crowded as the black holes will occupy a larger and larger fraction of the 3--sphere, resulting in increasing distortion of the space between the black holes as well as of the black holes themselves, both due to the presence of near companions. In fact, we have just proven that it is quite the opposite: for large $N$ the Universe becomes effectively more and more \emph{empty}, because an increasing part of the 3--sphere is taken up by the  region far away from the black holes  which asymptotically lies at their infinity. An observer  in this region no longer feels the gravitational field of any single black hole, but rather the collective influence of all
of them. The metric in this area is very close to the closed FLRW metric despite the fact that no ordinary matter is present.

All these properties can be observed easily  in figure  \ref{fig:3configs}, where we have plotted the shape of the conformal factor in three
first configurations from the infinite sequence (see also figure 1 in \cite{Clifton:2012qh} for a similar plot for the six regular configurations).

\bfi
\bce
\includegraphics[width=0.3\textwidth]{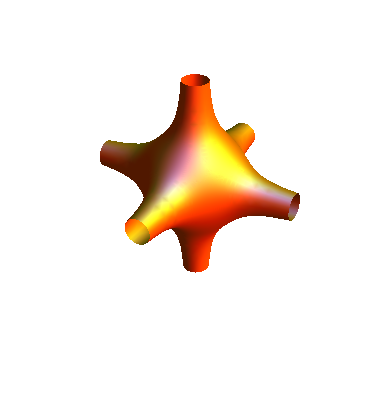}
\includegraphics[width=0.3\textwidth]{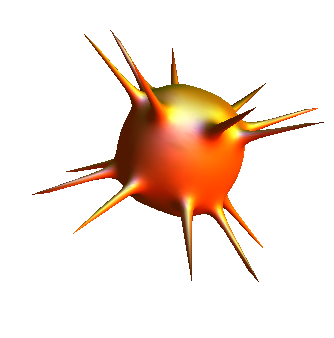}
\includegraphics[width=0.3\textwidth]{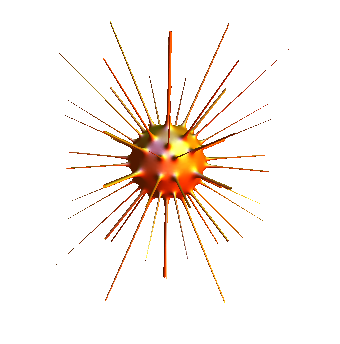}
\caption{Three--dimensional slices though the $H_8$, $H_{64}$ and $H_{216}$ configurations. The distance of every point on the surface the origin is proportional to $\Phi$ at the corresponding point in $S^3$. Note that as the number of punctures grows, the spikes, representing the singularities in $\Phi$ become more and more narrow. At the same time the space between the punctures becomes increasingly round.
\label{fig:3configs}}
\ece
\efi

Consider now the relative mass deficit $\sigma_M$. Theorem \ref{th:M}, in conjunction with the results from the previous subsection, yields
an estimate of the form of $\sigma_M \le O(N^{-1/3})$. Thus for very large $N$ we may simply neglect all the backreaction effects and plug into 
the Friedmann equation the energy density calculated in the most naive way, by adding up the ADM masses of all black holes and dividing the result by the effective volume of the Universe. This way we have proved that in $H_N$ the backreaction effect of
mass non--additivity vanishes in the continuum limit.
Both of these results agree with what has been observed for the regular configurations in \cite{Clifton:2012qh}.

\begin{remark}
As we mentioned in Section 2, in \cite{Clifton:2012qh, Bentivegna:2012ei, Clifton:2013jpa} the fitting and comparison with a homogeneous FLRW model was
done differently, by matching the length of the edge of a single big cell with the length of the edge from a corresponding lattice on a round 3--sphere. 
As we have shown, the value of the 3--metric on the edge converges to $\Mean\Phi^4 q_0$ relatively quickly because it lies as far as possible from the black holes. It follows easily that the values of $\eeff a$ and $\eeff M$ we obtain from such fitting converge to (\ref{eq:eeffa}) and (\ref{eq:Meff}) respectively. More generally, this convergence should hold for any reasonable definition of the effective parameters as long as the sequence in question has a continuum limit and the definition is sensitive only to the physical metric in the faraway region, avoiding the strong metric distortions near compact sources. This observation gives justification for the unphysical definition of $\eeff a$ via the averaging over the unphysical volume form: in the continuum limit it doesn't really matter which definition of the effective 3--metric we actually use to fit the FLRW model and calculate the effective mass, because as long as we only probe the metric in the faraway region they all yield asymptotically the same value for $N\to \infty$.
\end{remark}

\section{Numerical results}
\label{sec:numerical}

We will now present a few numerical results concerning the mass deficit and the convergence of the conformal factor towards the mean value. First let us consider the 6 regular configurations $R_N$ (see  \cite{Clifton:2012qh} for the details of their construction). Let $p$ be a vertex of the dual configuration, i.e. any point where the edges of a single lattice cell  intersect. These points lie at the largest possible distance from the black holes, so we may expect  a fast convergence of $\Phi$ at $p$ towards the mean value. The modified discrepancy of these configurations is difficult to evaluate even numerically, but can easily be estimated from above by the diameter of the corresponding regular tessellation of $S^3$ via Lemma \ref{lm:Ediam}. $\varepsilon$ is taken to be $0.1$ everywhere.

Figure \ref{fig:DeltaPhiplatonic} shows the convergence of $\sigma_\Phi$ defined by (\ref{eq:r}) towards 0. The bounds given by Theorem \ref{th:Phi} are plotted as well, separately for functions $U_{1,\varepsilon}$, $U_{2,\varepsilon}$ and $U_{3,\varepsilon}$. The estimate for $\sigma_\Phi$ is thus given by the largest of these functions, i.e. $U_{2,\varepsilon}$. As we can see, the actual value of  $\sigma_\Phi$  is 3 orders of magnitude smaller than the bounds. We note here that $\sigma_\Phi$ seems to be the same for the pairs $R_8$ and $R_{16}$ as well as $R_{120}$ and $R_{600}$. These configurations are dual in the sense that they arise from dual polytopes ($R_5$ and $R_{24}$ are self--dual). Indeed it is  relatively easy to show that for our choice of $p$   $\sigma_\Phi$   is the same
in dual configurations.
\bfi
\bce
\includegraphics[width=0.9\textwidth]{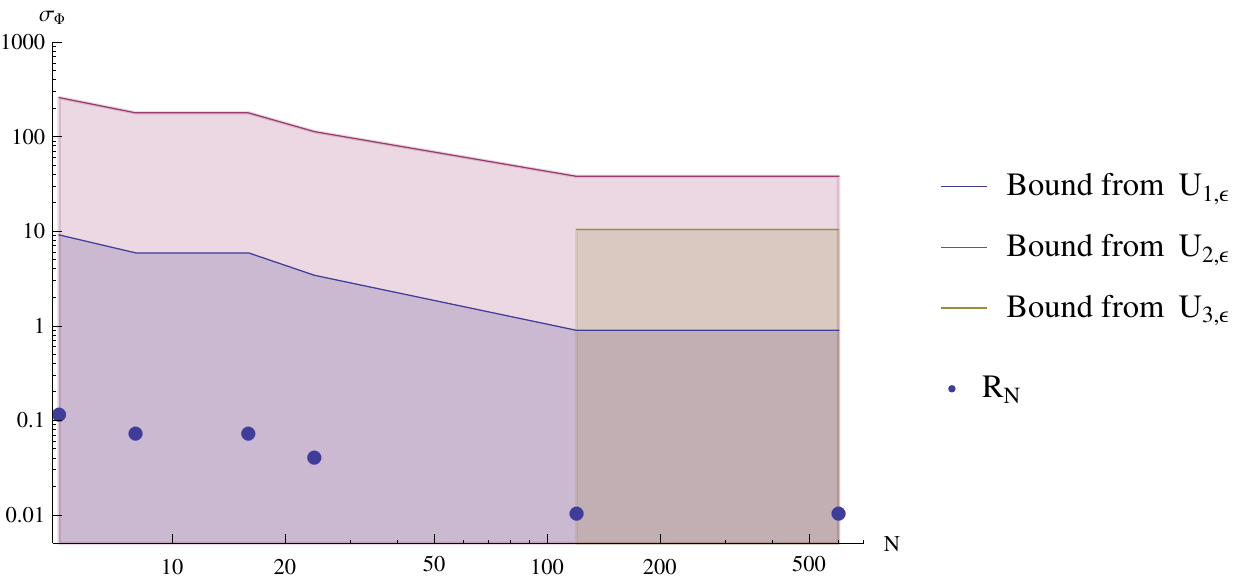}
\caption{Log-log plot of  $\sigma_\Phi$  for the 6 regular configurations and the bounds from Theorem \ref{th:Phi}}
\label{fig:DeltaPhiplatonic}
\ece
\efi
\bfi
\bce
\includegraphics[width=0.9\textwidth]{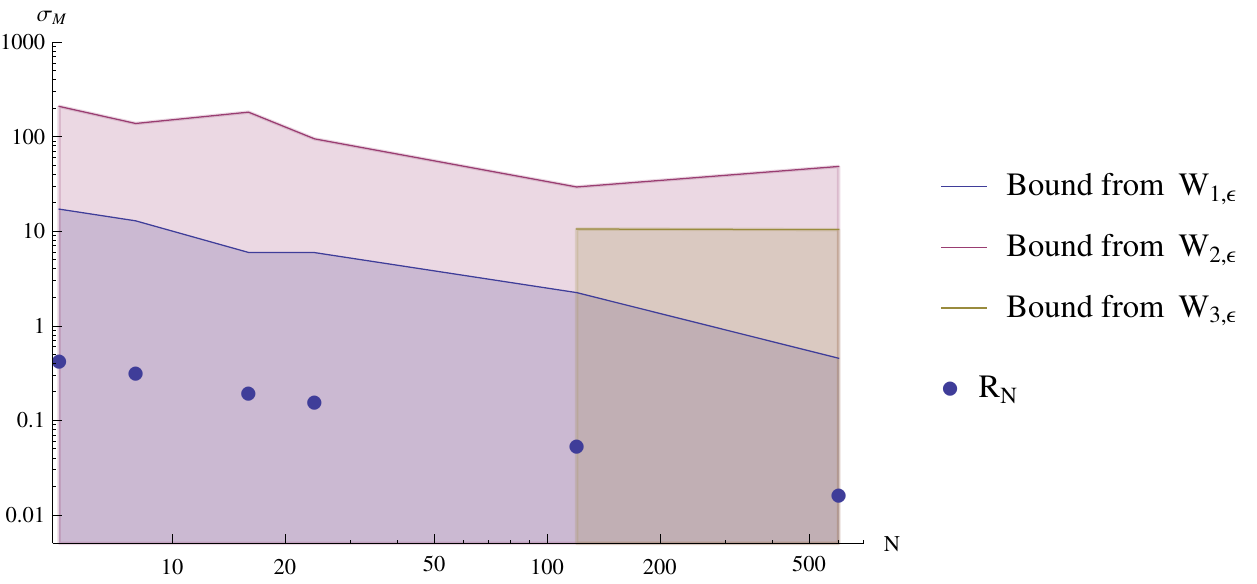}
\caption{Log-log plot of  $\sigma_M$ for the 6 regular configurations and the bounds from Theorem \ref{th:M}}
\label{fig:DeltaMplatonic}
\ece
\efi

In figure \ref{fig:DeltaMplatonic} we present the relative mass deficit for all 6 regular configurations. The value seems to
decrease with $N$ and once again it is orders of magnitude smaller than the bounds given by $W_{1,\varepsilon}$, $W_{2,\varepsilon}$
and $W_{3,\varepsilon}$. 
Note that for the first 4 configurations the conditions (\ref{eq:conditPhi}) and (\ref{eq:conditM})  are not satisfied. Thus $ U_{3,\varepsilon}$ and $W_{3,\varepsilon}$ are undefined and Theorems \ref{th:Phi} and \ref{th:M} do not apply to them. Nevertheless even in these cases the actual values of both $\sigma_M$ and $\sigma_\Phi$ are of orders of magnitude smaller than the  bounds by the remaining two functions.

In figures \ref{fig:DeltaPhiN8} and \ref{fig:DeltaMN8} we have plotted  $\sigma_\Phi$  and $\sigma_M$ for the first 6 $H_N$ configurations. Point $p$ was chosen to be one of the vertices of the initial tessellation into 8 big quasi-cubes, namely $(\frac{1}{2},\frac{1}{2},\frac{1}{2},\frac{1}{2})$ in terms of the Cartesian coordinates on $\Rr^4$. Conditions (\ref{eq:conditPhi}) and (\ref{eq:conditM}) are satisfied from the third element of the sequence up. We see again steady decrease of both $\sigma_\phi$ and $\sigma_M$. Linear dependence on a log--log plot suggests   a negative power convergence rate. 
\bfi
\bce
\includegraphics[width=0.9\textwidth]{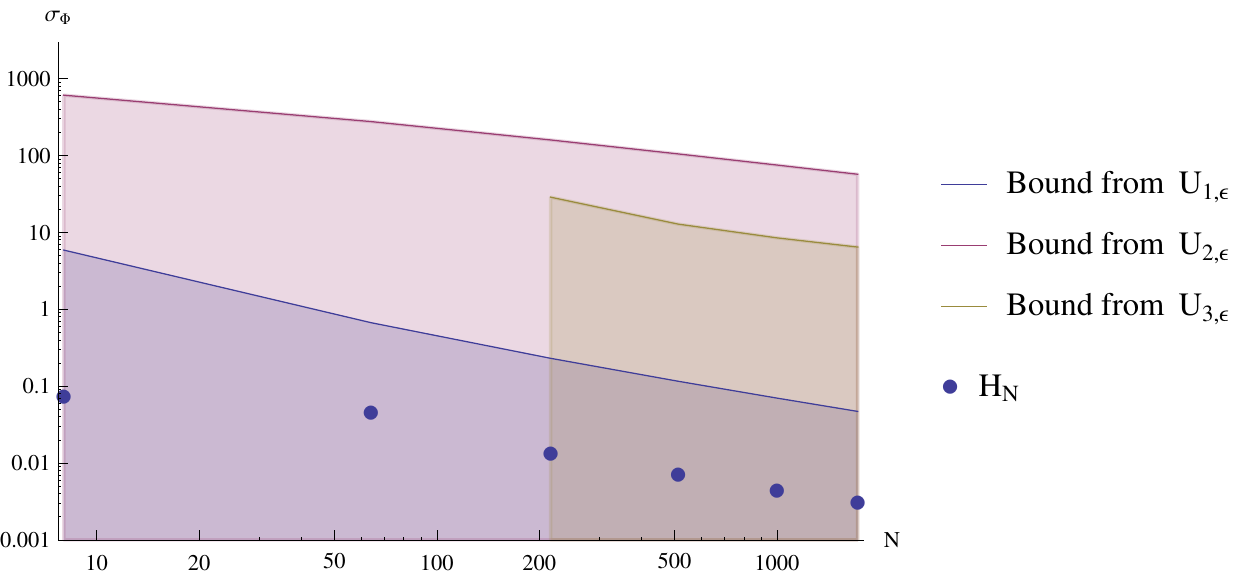}
\caption{Log-log plot of  $\sigma_\Phi$  for the  $H_N$ configurations and the bounds from Theorem \ref{th:Phi}}
\label{fig:DeltaPhiN8}
\ece
\efi
\bfi
\bce
\includegraphics[width=0.9\textwidth]{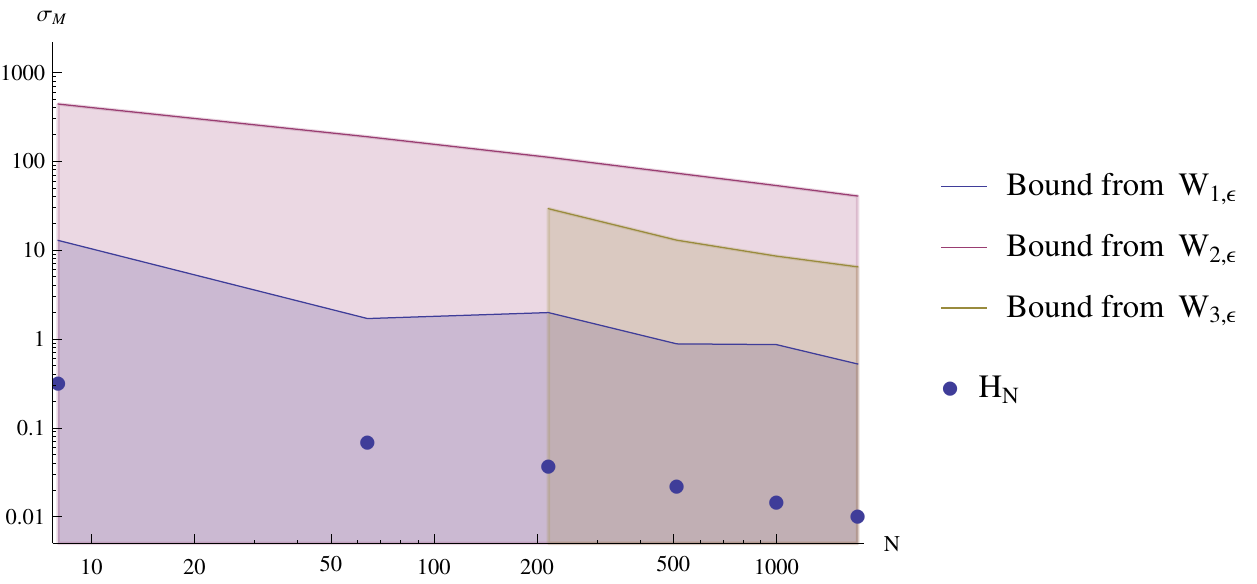}
\caption{Log-log plot of  $\sigma_M$  for the $H_N$ configurations and the bounds from Theorem \ref{th:M}}
\label{fig:DeltaMN8}
\ece
\efi
Obviously the inequalities in Theorems \ref{th:Phi} and \ref{th:M}  greatly overestimate both relative variation from the mean of $\Phi_p$ as well as
the relative mass difference. We may however ask if they correctly predict the exponent of the power convergence rate.

Recall that for the infinite sequence of configurations we have constructed we have $E = O(N^{-1/3})$, $\delta_{\min} = O(N^{-1/3})$, $\delta_{\max} = O(N^{-1/3})$. Moreover, since $p$ is located on the cell boundary in all configurations, 
$\lambda_{\min}=O(N^{-1/3})$ as we have shown before. If we substitute these relations to (\ref{eq:U}) and (\ref{eq:W}) we see that 
 $W_{\epsilon} = O(N^{-1/3})$ and $U_{\epsilon} = O(N^{-1/3})$ for all $\varepsilon$. It follows that the right hand side of both inequalities decreases like $O(N^{-1/3})$ as well.
 On the other hand
plots in figures \ref{fig:DeltaPhi23} and \ref{fig:DeltaMboth}, where we fit a straight line to the data using the least-squares method, suggest a faster  convergence with the exponent close to $-\frac{2}{3}$. 
\bfi
\bce
\includegraphics[width=0.9\textwidth]{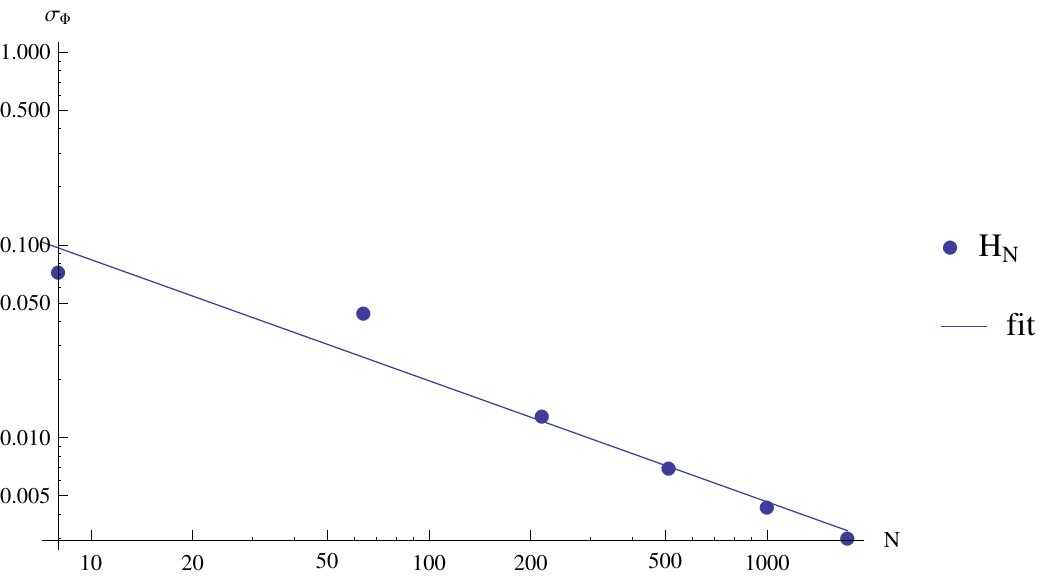}
\caption{$\sigma_\Phi$ for the $H_N$ configurations and the least squares fit. The slope of the straight line corresponds to $N^{-0.63}$}
\label{fig:DeltaPhi23}
\ece
\efi
\bfi
\bce
\includegraphics[width=0.9\textwidth]{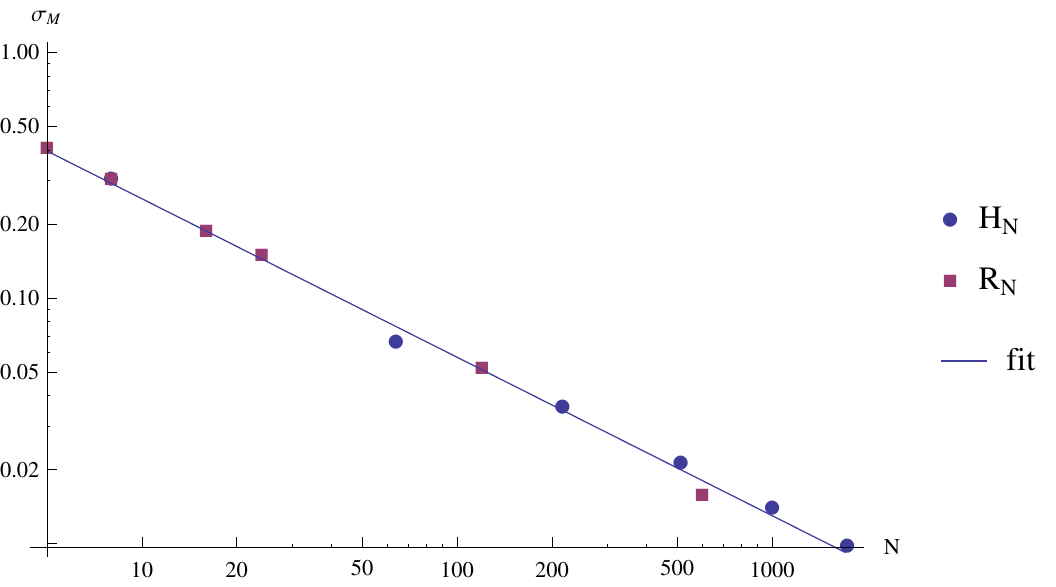}
\caption{$\sigma_M$ for both types of configurations and the least squares fit. The slope of the straight line  corresponds to $N^{-0.65}$}
\label{fig:DeltaMboth}
\ece
\efi

Finally let us look at the $H_N$ configurations with a slight modification. Namely, instead of placing a single black hole 
near the center of each quasi--cubic cell, we put there  a pair of them, with identical mass parameters and located $O(N^{-1})$ from each other.
This spacing means that $\lambda_{\min} = O(N^{-1})$ and the mass deficit is not guaranteed by Theorem \ref{th:M} to converge to $0$ even when $N\to \infty$. On the other hand \ref{th:Phi} still applies to points lying sufficiently far away from the centers of the small quasi--cubes and we expect that the metric tensor will tend to the FLRW form there. Indeed, as we see on figures \ref{fig:DeltaPhipairs} and \ref{fig:DeltaMpairs}, the  $\sigma_\Phi$  parameter
goes relatively quickly to 0 while the relative mass deficit \emph{grows} with $N$, approaching the value of approximately $1$.
\bfi
\bce
\includegraphics[width=0.9\textwidth]{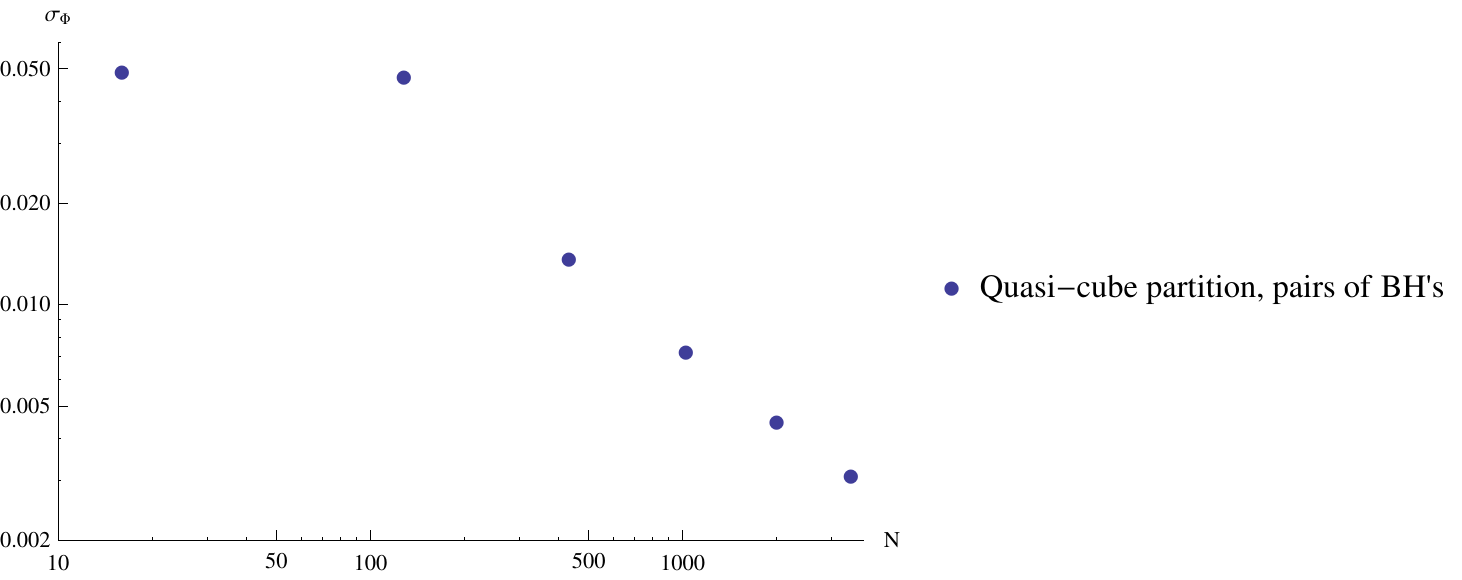}
\caption{$\sigma_\Phi$ for evenly distributed close pairs  of black holes}
\label{fig:DeltaPhipairs}
\includegraphics[width=0.9\textwidth]{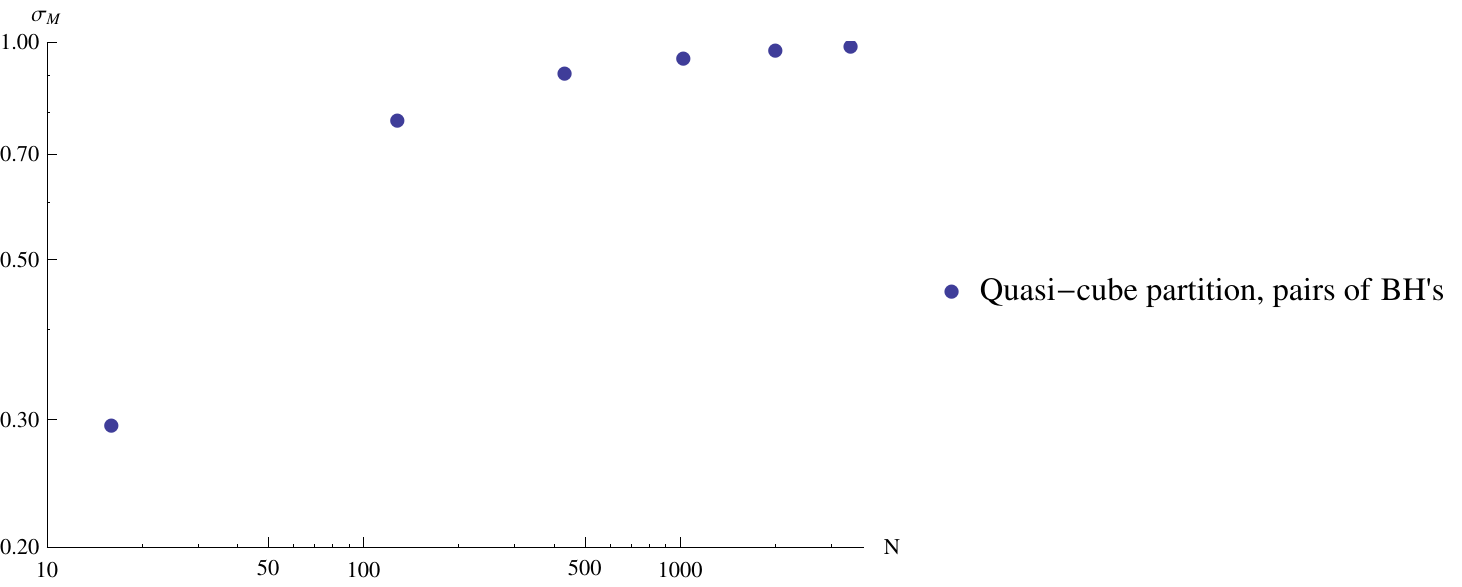}
\caption{$\sigma_M$ for evenly distributed close pairs of black holes}
\label{fig:DeltaMpairs}
\ece
\efi

This example shows explicitly that the existence of the continuum limit does not imply automatically the vanishing of backreaction when $N\to \infty$. Indeed, black hole clustering in this example results in big relativistic corrections to the total mass of each black hole pair even as we pass to the continuum limit. Even though the total masses of the pairs measured far away are additive in the continuum limit, the masses of their constituents are not. This highlights the necessity of a condition limiting the black hole clustering in the inequality for backreaction: backreaction simply does not go to 0 if we allow the black holes to come too close to each other.

\section{Summary and conclusions}
We have discussed  the time--symmetric $S^3$ vacuum initial data with black holes as the only source of the gravitational field, focusing on their properties  when the number of black holes is large. The initial data is constructed via the Lichnerowicz--York method
from the metric of a 3--sphere by multiplication with an explicitly known conformal factor. The factor function contains a fixed number
of singularities (punctures), with arbitrary positive leading coefficients, called mass parameters, and  located at arbitrarily chosen points of $S^3$. These singularities correspond to  throats leading via a minimal surface to an
asymptotically flat region, resembling strongly  Schwarzschild black holes. Apart from the black holes the model contains  no other types of matter.

 We have proved that as the number of black holes $N$ goes to infinity  the spatial metric tensor in this model approaches  asymptotically 
the uniform metric tensor of a closed FLRW model everywhere except the regions located in the immediate vicinity of the  black holes, provided that the black holes cover the uniformly the 3--sphere. This happens despite the fact that there is no continuous, constant matter density in this model. On the other hand, in the regions contained within a finite number of Schwarzschild radii  from a black hole  the metric tensor is strongly affected by its presence and  approaches asymptotically the metric of an isolated Schwarzschild black hole rather than an FLRW model. We have proved however that as we increase $N$ the total size of this near--black--hole region,
measured with respect to the 'round' metric,  decreases quickly and asymptotically almost all of the 3--sphere becomes FLRW. We thus establish the existence of the continuum limit for these models  as a kind of pointwise limit with an excluded region.

The main tool we used was an inequality providing the upper bound for the relative difference between the conformal factor of the initial data at a point $p$ and the mean conformal factor. Since the spatial metric $q$ is conformally equivalent to $S^3$, the inequality measures the difference between $q$ and the homogeneous, averaged metric. The upper bound is a function of the distance from the nearest puncture $\lambda_{\min}$ and the modified cap discrepancy $E$, a quantity characterizing  the distribution of the black holes. $E$ measures the difference between the standard measure of volume on $S^3$ and the discretized measure centered at the punctures. It is
small if and only if the distribution of black hole mass is uniform around the sphere, i.e. the black holes do not gather 
in some regions of $S^3$ leaving other ones empty. The bounding function is  non--increasing in $\lambda_{\min}$ and
non--decreasing in $E$ and it goes to 0 if $E$ goes to $0$ sufficiently faster than $\lambda_{\min}$. 

We have also proved another inequality bounding the difference between the effective mass of the model, calculated from the 
mean conformal factor via the Einstein equations, and the sum of the ADM masses of the individual black holes. Its vanishing for large $N$ would mean 
that asymptotically the non--linearity of general relativity plays no role and we may assume simple mass additivity in these models for sufficiently large $N$.
The estimate we have proved is again a function of the modified cap discrepancy $E$, the ratio of the largest mass parameter to the sum of all 
mass parameters, the largest distance from a puncture and the nearest other one  and finally the smallest distance between any of the two punctures. The dependence is non--increasing in the first three variables, but non--decreasing in the last one. 
The inequality allows to prove that the two total masses coincide asymptotically if the black holes do not cluster, i.e. the distance
between the two nearest of them does not decrease too fast as we increase $N$. 

The estimates in the inequalities above depend on  certain parameters of the microscopic matter distribution. It is however  not clear 
a priori the inequalities are indeed useful, i.e. if there exist configurations for which the their right hand sides can be made as small as possible for large $N$. Therefore we also presented an explicit construction of an infinite sequence of configurations in which both the continuum limit and mass additivity are attained
in the sense outlined above. The construction is based on partitioning the 3--sphere into a number of quasi--cubic domains and 
placing a black hole into the center of each of them. 

We finally presented the effect of numerical study of the configurations constructed above as well as the regular ones. It turned
out that the inequalities in question grossly overestimate both the deviation of the conformal factor from the mean value and the relative mass deficit.
We have also presented an example of a sequence of configurations made of close pairs of black holes for which the continuum limit seems to
be attained, but there remains a residual difference between the effective total mass and the sum of the black hole masses even
as $N\to \infty$.

The most important general lessons to be learned from the examples we have discussed  can be summarized as follows:
\begin{enumerate}
\item Sequences of solutions of Einstein equations with discretized matter sources may have a  continuum limit in which
the metric tends pointwise to the metric of corresponding to a continuous matter distribution as
the number of sources goes to infinity.
\item For very compact sources, such as black holes, in which the gravitation field is very strong, the continuum limit is not attained everywhere. The typical scenario looks as follows: for any finite number of objects there exist  regions very close to the sources where the metric is very 
strongly distorted by their presence and therefore never approaches the continuous one. Nevertheless these regions  shrink quickly as
the number of sources goes to infinity and eventually most of the spacetime consists of the far--away region where the influence 
of any single source is negligible. In this region the gravitational field is determined by the collective influence of all discrete sources and the metric tensor tends to the one with the continuous matter source provided that the distribution of masses agrees with the corresponding continuous energy distribution.
\item The mere existence of the continuous limit does not guarantee that the effects of backreaction, arising due to the non--linearity of GR, are negligible. It is necessary
to assume something more about the microscopic distribution of sources to make sure that the effective stress--energy tensor of the continuous distribution asymptotically agrees with the one we obtain from simply adding up all the microscopic contributions. In our
case we needed to assume that the distance between the two nearest black holes does not go to 0 too fast.
\item Whilst the backreaction effects in GR may be very difficult to evaluate analytically, it is still possible to prove rigorously inequalities
estimating them given an appropriate set of parameters characterizing the microscopic distribution of matter.
\end{enumerate}

The numerical results from Section \ref{sec:numerical} suggest that  the estimates we have proven, while useful, are far from being optimal. In particular, the dependence on the smallest distance between the black holes seems quite restrictive: normally we would expect that a single pair of black holes forming a bound system does not spoil the additivity for large $N$. Therefore there is probably a large room for improvement in
Theorems \ref{th:Phi} and \ref{th:M}.
Another interesting direction for future research would be the  generalization of the results  to more complicated initial data, with arbitrary distribution of matter and no assumptions about the symmetry of the coarse--grained metric. An even more interesting and  difficult  generalization  would include the full dynamics of the system and provide similar estimates for the rest of the
backreaction terms, i.e. the
additional pressures in the effective stress--energy tensor arising from the discrete nature of the sources.

\section*{Acknowledgements}
The author would like to thank the Max Planck Institute for Gravitational Physics--Albert Einstein Institute in Potsdam for hospitality. The work was supported by the project \emph{
``The role of small-scale inhomogeneities in general relativity and cosmology''} (HOMING PLUS/2012-5/4), realized within the Homing Plus programme of Foundation for Polish Science, co--financed by the European Union from the Regional Development Fund.

\appendix
\section{Proof of Lemma \ref{lm:Hlawka-Zaremba}} \label{app:Hlawka-Zaremba}

 We assume that $0<\lambda_1\le\lambda_2\le\dots\lambda_N\le\pi$, relabeling the punctures if needed. Function $\disc(\lambda)$ is discontinuous
at $\lambda_n$ and continuous everywhere else. Thus 
we can perform integration by parts on the right hand side on all intervals $(0,\lambda_1), (\lambda_1,\lambda_2),\dots, (\lambda_N,\pi)$. The derivative of
$\disc(\lambda)$ inside these intervals is equal to the derivative of $\Xi(\lambda)$, which is $\sin^2\lambda$. We get:
\bea
\textrm{rhs} &=& \frac{2}{\pi}\int_0^{\lambda_1} f(\lambda)\,\sin^2\lambda\,\dd\lambda -
 \disc(\lambda)\,f(\lambda) \Big|_{0}^{\lambda_1} \nonumber\\
 &&+ \frac{2}{\pi}\int_{\lambda_1}^{\lambda_2} f(\lambda)\,\sin^2\lambda\,\dd\lambda -
 \disc(\lambda)\,f(\lambda) \Big|_{\lambda_1}^{\lambda_2} + \cdots\ \nonumber\\
 &&+ \frac{2}{\pi}\int_{\lambda_N}^{\pi} f(\lambda)\,\sin^2\lambda\,\dd\lambda -
 \disc(\lambda)\,f(\lambda) \Big|_{\lambda_N}^{\pi}\nonumber
\eea
The boundary term at $\pi$ vanishes because $\disc(\lambda)$ vanishes there and $f$ is regular. Since $\disc(\lambda)$ is of the order of $\lambda^3$ near zero, while $f$ is by assumption of the order higher than ${-3}$, their product vanishes in the limit of $\lambda\to 0$ as well. The remaining boundary terms at $\lambda_n$ do not cancel each other out because of the 
discontinuity of $\disc(\lambda_n)$ at these points. A pair at each $\lambda_n$ contributes $-\alpha_n/\alpha$, while the remaining integrals add up to the integral over
the whole $(0,\pi)$:
\bea
\textrm{rhs} &=& \frac{2}{\pi}\int_0^{\pi} f(\lambda)\,\sin^2\lambda\,\dd\lambda - \sum_{n=1}^N \frac{\alpha_n}{\alpha}\,f(\lambda_n) = \textrm{lhs}.\nonumber
\eea
\qed

\section{Relation between the distances and volumes on $\Rr^3$ and $S^3$ induced by projection $\Psi$} \label{app:distances}

Fix mapping (\ref{eq:Psiproj}) in the big quasi--cube around $(1,0,0,0)$ and introduce a spherical coordinate system $(\lambda,\vartheta,\varphi)$  with the $(1,0,0,0)$ as the
pole $\lambda = 0$. The mapping now takes the form of
\bea
\Psi: S^3 \ni (\lambda,\vartheta,\varphi) \mapsto (r(\lambda) = \tan \lambda,\vartheta,\varphi)\nonumber
\eea
with $(r,\vartheta,\varphi)$ being the standard spherical coordinate system on $\Rr^3$.
Now the pullback of the flat metric $h$ from $\Rr^3$ to $S^3$ takes the form of
\bea
\Psi^* h = \cos^{-4}\lambda \,\dd\lambda^2 + \tan^2\lambda(\dd\vartheta^2 + \sin^2\vartheta\,\dd\varphi^2). \nonumber
\eea
Consider the coframe $e^1 = \dd\lambda$, $e^2 = \sin\lambda\,\dd\vartheta$, $e^3 = \sin\lambda\sin\vartheta\,\dd\varphi$,
orthonormal with respect to $q_0$:
\bea
q_0 = e^1\otimes e^1 + e^2\otimes e^2 + e^3 \otimes e^3. \label{eq:q0e}
\eea 
The pullback can be expressed in terms of the coframe:
\bea
\Psi^*h = \cos^{-4}\lambda \,e^1\otimes e^1 + \cos^{-2}\lambda\,\left(e^2\otimes e^2 + e^3 \otimes e^3\right). \label{eq:he}
\eea 
 Note that $\lambda$ in a single quasi--cube satisfies $0 \le \lambda \le \frac{\pi}{3}$,  with the upper bound attained at the 8 vertices of the quasi--cube at $\left(\frac{1}{2},\pm \frac{1}{2},\pm \frac{1}{2},\pm \frac{1}{2}\right)$. Now comparing (\ref{eq:q0e})
  with (\ref{eq:he}) we see that the largest difference between $q_0(X,X)$ and $\Psi^*h(X,X)$ is attained
  at the vertices and when the vector $X$ in question is proportional $\partial_\lambda$. In this situation $  q_0(X,X)  = \frac{1}{16} \,\Psi^*h(X,X)$. 
  On the other hand the metrics coincide at $\lambda=0$. Altogether,  we have
  \bea
  \frac{1}{16} \,\Psi^*h(X,X) \le q_0(X,X) \le \Psi^*h(X,X)\nonumber
  \eea 
for all points and tangent vectors inside the big quasi-cube.
This inequality carries over to the lengths of the curves: for any curve $\gamma(t)$ we have
\bea
 \frac{1}{4}\int_\gamma \sqrt{\Psi^*h(\dot\gamma,\dot\gamma)} \dd t \le \int_\gamma \sqrt{q_0(\dot\gamma,\dot\gamma)} \dd t
\le \int_\gamma \sqrt{\Psi^*h(\dot\gamma,\dot\gamma)} \dd t\nonumber
\eea
 
Recall that the distances $D(\cdot,\cdot)$ and $\Lambda(\cdot,\cdot)$ is the length of the shortest curve between two points
measured with $h$ and $q_0$ respectively. For a pullback of a straight line $\gamma_1(t)$, minimizing the length measured by $\Psi^*h$ between two points $x$ and $y$, and great circle $\gamma_2(t)$, minimizing
the length measured by $q_0$  we have
\bea
\frac{1}{4}D(x,y) &=& \frac{1}{4}\int_{\gamma_1} \sqrt{\Psi^*h(\dot\gamma_1,\dot\gamma_1)} \,\dd t \le \frac{1}{4}\int_{\gamma_2}\sqrt{\Psi^*h(\dot\gamma_2,\dot\gamma_2)} \,\dd t \nonumber\\
&\le& \int_{\gamma_2} \sqrt{q_0(\dot\gamma_2,\dot\gamma_2)} \,\dd t = \Lambda(x,y).\nonumber
\eea
Swapping the role of $\gamma_1(t)$ and $\gamma_2(t)$ in the reasoning above we prove easily that $\Lambda(x,y) \le D(x,y)$
and hence (\ref{eq:DLD}).

Finally let us estimate the spherical volumes. Let $\kappa = \dd x^2\wedge \dd x^3 \wedge \dd x^4$ be the standard volume form on $\Rr^3$. Note that $\eta = e^1\wedge e^2\wedge e^3$ and $\Psi^*\kappa = \cos^{-3}\lambda\,e^1\wedge e^2\wedge e^3$.
A similar line of reasoning as before shows that
\bea
\frac{1}{8}\Psi^*\kappa(X,Y,Z) \le \eta(X,Y,Z) \le \Psi^*\kappa(X,Y,Z) . \nonumber
\eea
for any properly oriented triple of vectors $X,Y,Z$. After integrating over any domain $C$ of proper orientation we obtain upper and lower bounds for the spherical volume  in terms of the "flat" volume:
\bea
\frac{1}{8}\vol_{\Rr^3}(C) \le \vol(C)  \le \vol _{\Rr^3}(C) \label{eq:volumes}. 
\eea

\section*{References}
\bibliographystyle{iopart-num}
\bibliography{s3-limit}

 \end{document}